\documentclass[11pt]{llncs}
\usepackage{amsmath,amssymb}
\usepackage{graphicx}
\newcommand{\comm}[1]{}

\usepackage{fullpage}
\usepackage[top=0.8in, bottom=0.9in, left=0.9in, right=0.9in]{geometry}

\usepackage{color,soul}
\usepackage{amsmath,amssymb}	
\usepackage[euler]{textgreek}
\usepackage{algorithm}
\usepackage{algpseudocode}
\usepackage{hyperref,cite}

\usepackage{lineno}
\def\calI{\mathcal{I}}
\def\calH{\mathcal{H}}
\def\calA{\mathcal{A}}
\newtheorem{observation}{Observation}
\begin{document}

\title{Geometric Hitting Set for Line-Constrained Disks and Related Problems\thanks{A preliminary version of this paper will appear in {\em Proceedings of the 18th Algorithms and Data Structures Symposium (WADS 2023)}. This research was supported in part by NSF under Grants CCF-2005323 and CCF-2300356.}}
\author{Gang Liu %\inst{1}
\and
Haitao Wang %\inst{1}
}

 \institute{
 School of Computing\\
  University of Utah, Salt Lake City, UT 84112, USA\\
  \email{u0866264@utah.edu, haitao.wang@utah.edu}
}

\maketitle

\pagestyle{plain}
\pagenumbering{arabic}
\setcounter{page}{1}

%\vspace{-0.2in}
\begin{abstract}
Given a set $P$ of $n$ weighted points and a set $S$ of $m$ disks in the plane, the hitting set problem is to compute a subset $P'$ of points of $P$ such that each disk contains at least one point of $P'$ and the total weight of all points of $P'$ is minimized. The problem is known to be NP-hard. In this paper, we consider a line-constrained version of the problem in which all disks are centered on a line $\ell$. We present an $O((m+n)\log(m+n)+\kappa \log m)$ time algorithm for the problem, where $\kappa$ is the number of pairs of disks that intersect.
%Alternatively, we can also solve the problem in $O(nm\log(m + n))$ times.
For the unit-disk case where all disks have the same radius, the running time can be reduced to $O((n + m)\log(m + n))$. In addition, we solve the problem in $O((m + n)\log(m + n))$ time in the $L_{\infty}$ and $L_1$ metrics, in which a disk is a square and a diamond, respectively.
%Besides, we propose an easier and more straightforward algorithm in unweighted case for 1D, $L_{\infty}$, $L_2$ disk metrics.
Our techniques can also be used to solve other geometric hitting set problems. For example, given in the plane a set $P$ of $n$ weighted points and a set $S$ of $n$ half-planes, we solve in $O(n^4\log n)$ time the problem of finding a minimum weight hitting set of $P$ for $S$. This improves the previous best algorithm of $O(n^6)$ time by nearly a quadratic factor. 
\end{abstract}
%\vspace{-0.1in}

%\vspace{-0.2in}
\section{Introduction}
\label{sec:intro}

Let $S$ be a set of $m$ disks and $P$ be a set of $n$ points in the plane such that each point of $P$ has a weight. The {\em hitting set problem} on $S$ and $P$ is to find a subset $P_{opt} \subseteq P$ of minimum total weight so that each disk of $S$ contains a least one point of $P_{opt}$ (i.e, each disk is {\em hit} by a point of $P_{opt}$). The problem is NP-hard even if all disks have the same radius and all point weights are the same~\cite{ref:DurocherDu15,ref:KarpRe72,ref:MustafaIm10}.

In this paper, we consider the {\em line-constrained} version of the problem in which centers of all disks of $S$ are on a line $\ell$ (e.g., the $x$-axis). To the best of our knowledge, this line-constrained problem was not particularly studied before. We give an algorithm of $O((m+n)\log(m+n)+\kappa\log m)$ time, where $\kappa$ is the number of pairs of disks that intersect. We also present an alternative algorithm of  $O(nm\log(m+n))$ time. For the {\em unit-disk case} where all disks have the same radius, we give a better algorithm of $O((n + m)\log(m + n))$ time. We also consider the problem in $L_{\infty}$ and $L_1$ metrics (the original problem is in the $L_2$ metric), where a disk becomes a square and a diamond, respectively; we solve the problem in $O((m + n)\log(m + n))$ time in both metrics. The 1D case where all disks are line segments can also be solved in $O((m + n)\log(m + n))$ time.
In addition, by a reduction from the element uniqueness problem, we prove an $\Omega((m+n)\log(m+n))$ time lower bound in the algebraic decision tree model even for the 1D case (even if all segments have the same length and all points of $P$ have the same weight). The lower bound implies that our algorithms for the unit-disk, $L_{\infty}$, $L_1$, and 1D cases are all optimal.

Hitting set is a fundamental problem that has attracted much attention in the literature.
Many variations of the problem are intractable. One motivation of our study is to see
to what extent the hitting set problem has efficient algorithms and how efficient we
could make. Our problem may also have direct applications. For example, suppose a number of sensors are deployed along a line (e.g., a high way or a rail way); 
%The barrier coverage problem is to place sensors on the
%boundary of a region (say, a polygon) to detect intruders to get inside the region.
%After sensors are placed, 
we need to determine locations to build base stations to
communicate with sensors (each sensor has a range for communications and the base
station must be located within its range). How to determine a minimum number of locations for base stations is exactly an instance of our problem.

Although the problems are line-constrained, our techniques can be utilized to solve various other geometric hitting set problems. If all disks of $S$ have the same radius and the set of disk centers is separated from $P$ by a line $\ell$, the problem is called {\em line-separable unit-disk hitting set problem}. 
Our algorithm for the line-constrained $L_2$ general case can be used to solve the problem in $O(nm\log (m+n))$ time or in $O((m+n)\log (m+n)+\kappa\log m)$ time, where $\kappa$ is the number of pairs of disks whose boundaries intersect in the side of $\ell$ that contains $P$.
Interestingly, we can also employ the algorithm to tackle the following
{\em half-plane hitting set problem}. Given in the plane a set $S$ of $m$ half-planes and a set $P$ of $n$ weighted
points, find a subset of $P$ of minimum total weight so that each  half-plane of $S$ contains at least one point of the subset. For the {\em lower-only case} where all half-planes are
lower ones, Chan and Grant~\cite{ref:ChanEx14} gave an $O(m^2n(m+n))$
time algorithm. Notably, recognizing that a half-plane can be seen as a special unit disk with an infinite radius, our line-separable unit-disk
hitting set algorithm can be applied to solve the problem
in $O(nm\log (m+n))$ time or in $O(n\log n+m^2\log m)$ time. This
improves the result of~\cite{ref:ChanEx14} by nearly a quadratic
factor. 
%(note that the techniques of~\cite{ref:ChanEx14} are applicable to more general problem settings such as downward shadows of $x$-monotone curves). 
For the general case where both upper and lower half-planes
are present, Har-Peled and Lee~\cite{ref:Har-PeledWe12} proposed an
algorithm of $O(n^6)$ time when $m=n$. Based on observations, we manage to reduce the problem to $O(n^2)$ instances of the lower-only case problem and consequently solve the problem in $O(n^3m\log (m+n))$ time or in
$O(n^3\log n+n^2m^2\log m)$ time using our lower-only case algorithm. The runtime is $O(n^4\log n)$ when $m=n$, which  improves the one
in~\cite{ref:Har-PeledWe12} by nearly a quadratic factor.
We believe that our techniques have the potential to find numerous other applications.

%\subsection{Related work}

\paragraph{\bf Related work.}
The hitting set and many of its variations are fundamental and have been studied
extensively; the problem is usually hard to solve, even
approximately~\cite{ref:MorenoTh13}. Hitting set problems in geometric settings
have also attracted much attention and most problems are NP-hard,
e.g.,\cite{ref:BusPr18,ref:ChanEx14,ref:EvenHi05,ref:GanjugunteGe11,ref:MustafaPt09}, and
some approximation algorithms are known~\cite{ref:EvenHi05,ref:MustafaPt09}.

A ``dual'' problem to the hitting set problem is the coverage problem. For our
problem, we can define its {\em dual coverage problem} as follows. Given a set
$P^*$ of $n$ weighted disks and a set $S^*$ of $m$ points, the problem is to find a
subset $P^*_{opt} \subseteq P^*$ of minimum total weight so that each point of
$S^*$ is covered by at least one disk of $P^*_{opt}$. This problem is also
NP-hard~\cite{ref:FederOp88}. The line-constrained problem where disks of $P^*$ are all centered on the $x$-axis was studied before
and polynomial time algorithms were proposed by Pedersen and Wang~\cite{ref:PedersenAl22}.
The time complexities of the algorithms of~\cite{ref:PedersenAl22}
match our results in this paper. Specifically, an algorithm of
$O((m+n)\log(m+n)+\kappa^*\log n)$ time was given in~\cite{ref:PedersenAl22} for
the $L_2$ metric, where $\kappa^*$ is the number of pairs of
disks that intersect~\cite{ref:PedersenAl22}; the unit-disk, $L_{\infty}$, $L_1$, and 1D cases were all solved in $O((n + m)\log(m + n))$ time~\cite{ref:PedersenAl22}. Other variations
of line-constrained coverage have also been studied,
e.g.,~\cite{ref:AltMi06,ref:BiloGe05,ref:PedersenOn18}.

The $L_2$ case coverage algorithm of~\cite{ref:PedersenAl22} can also be employed to solve  in $O((m+n)\log (m+n)+\kappa^*\log m)$ time the {\em line-separable unit-disk coverage problem} (i.e., all disks of $P^*$ have the same radius and the disk centers are separated from points of $S^*$ by a line $\ell$), where $\kappa^*$ is the number of pairs of disks whose boundaries intersect in the side of $\ell$ that contains $S^*$. Notice that since all disks are unit disks, we can reduce our line-separable unit-disk hitting set problem on $P$ and $S$ to the line-separable unit-disk coverage problem. Indeed, for each point $p$ of $P$, we create a {\em dual unit disk} centered at $p$ (with unit radius); for each disk of $S$, we consider its center as its {\em dual point}. As such, we obtain a set $P^*$ of $n$ dual disks and a set $S^*$ of $m$ dual points. It is not difficult to see that an optimal solution to the hitting set problem on $P$ and $S$ corresponds to an optimal solution to the coverage problem on $P^*$ and $S^*$. Applying the above coverage algorithm of \cite{ref:PedersenAl22} can solve our hitting set problem on $P$ and $S$ in $O((m+n)\log (m+n)+\kappa^*\log m)$ time, where $\kappa^*$ is the number of pairs of dual disks of $P^*$ whose boundaries intersect in the side of $\ell$ that contains $S^*$. Note that the time complexity is not the same as our above $O((m+n)\log (m+n)+\kappa\log m)$ time algorithm, 
where $\kappa$ is the number of pairs of disks of $S$ whose boundaries intersect in the side of $\ell$ that contains $P$, because $\kappa$ may not be the same as $\kappa^*$; indeed, $\kappa=O(m^2)$ while $\kappa^*=O(n^2)$. 

Using their algorithm for the line-separable unit-disk coverage problem, the lower-only half-plane coverage problem is solvable in $O(n\log n+m^2\log m)$~\cite{ref:PedersenAl22}, for $n$ points and $m$ lower half-planes. As above, by duality, we can also reduce our lower-only half-plane hitting set problem to the coverage problem and obtain an $O(m\log m+n^2\log n)$ time algorithm. Note again that this time complexity is not the same as our $O(n\log n+m^2\log m)$ time result, although they become identical when $m=n$. In addition, the general half-plane coverage problem, where both upper and lower half-planes are present, was also considered in \cite{ref:PedersenAl22} and an $O(n^3\log n+n^2m^2\log m)$ time algorithm was given. It should be noted that since both upper and lower half-planes are present, we cannot reduce the hitting set problem to the coverage problem by duality as in the lower-only case. Therefore, solving our general half-plane hitting set problem needs different techniques.

%The hitting set problem arises in a wide range of applications, such as emergency facility location, robotics, sensor networks, VLSI design, etc. \cite{ref:MustafaPt09,ref:hefeeda2007randomized,ref:eisner2012transit,ref:storandt2015approximation,ref:ganjugunte2011geometric} Although the problems in this paper is line-constrained version, it still potentially have applications, such as many model-based diagnostic methods can break the computational bottleneck with the help of our methods. \cite{ref:de2011hitting,ref:abreu2009low,ref:8892637}  For example, people want to build some power diagnostic centers which is represented by $P$ to provide fault diagnosis for some Power Distribution Network along a route which is represented by $\ell$. In addition, fault diagnosis can also be performed on the magnetic field range formed by these Power Distribution Network. So that we can refer to the solution to 1.5D metric in this article, according to Maxwell's equations, the difference in the magnitude of the electric power causes the difference in the range of the magnetic field change.

\paragraph{\bf Our approach.}
To solve the line-constrained hitting set problem, we propose a novel and interesting method, dubbed {\em dual transformation}, by reducing the hitting set
problem to the 1D dual coverage problem and consequently solve it by applying the 1D dual coverage algorithm of~\cite{ref:PedersenAl22}. Indeed, to the best of our knowledge, we are not aware of such a dual transformation in the literature. Two issues arise for this approach: The first one is to prove a good upper bound on the number of segments in the 1D dual coverage problem and the second is to compute these segments efficiently.
These difficulties are relatively easy to overcome for the 1D, unit-disk, and $L_1$ cases. The challenge, however, is in the $L_{\infty}$ and $L_2$ cases. Based on many interesting observations and techniques,
we prove an $O(n+m)$ upper bound and present an $O((n+m)\log (n+m))$ time algorithm to compute these segments for the $L_{\infty}$ case; for the $L_2$ case, we prove an $O(m+\kappa)$ upper bound and derive an $O((n+m)\log (n+m)+\kappa\log m)$ time algorithm.

\paragraph{\bf Outline.} The rest of the paper is organized as follows.
In Section~\ref{sec:pre}, we define notation and some concepts.
Section~\ref{sec:dualtransformation} introduces the dual transformation and solves the
1D, unit-disk, and $L_1$ cases.
Algorithms for the $L_{\infty}$ and $L_2$ cases are presented in
Sections~\ref{sec:infinity} and \ref{sec:l2}, respectively.
We discuss the line-separable unit-disk case and the half-plane hitting set problem in Section~\ref{sec:line-separable}.
Section~\ref{sec:conclusion} concludes the paper with a lower bound proof.
%$\Omega((m+n)\log (m+n))$ times lower bound by a reduction from the element uniqueness problem.

\section{Preliminaries}
\label{sec:pre}

We follow the notation defined in Section~\ref{sec:intro}, e.g., $P$, $S$, $P_{opt}$, $\kappa$, $\ell$, etc.  In this section, unless otherwise stated, all statements, notation, and concepts are applicable for all three metrics, i.e., $L_1$, $L_2$, and $L_{\infty}$, as well as the 1D case. Recall that we assume $\ell$ is the $x$-axis, which does not lose generality for the $L_2$ case but is special for the $L_1$ and $L_{\infty}$ cases. 

%We assume that $\ell$ is the $x$-axis.
We assume that all points of $P$ are above or on $\ell$ since if a point $p\in P$ is below $\ell$, we could replace $p$ by its symmetric point with respect to $\ell$ and this would not affect the solution as all disks are centered at $\ell$. For ease of exposition, we make a general position assumption that no two points of $P$ have the same $x$-coordinate and no point of $P$ lies on the boundary of a disk of $S$ (these cases can be handled by standard perturbation techniques~\cite{ref:EdelsbrunnerSi90}).
We also assume that each disk of $S$ is hit by at least one point of $P$ since otherwise there would be no solution (we could check whether this is the case by slightly modifying our algorithms).

For any point $p$ in the plane, we use $x(p)$ and $y(p)$ to refer to its $x$- and $y$-coordinates, respectively.

We sort all points of $P$ in ascending order of their $x$-coordinates; let $\{p_1, p_2,\cdots, p_n\}$ be the sorted list.
%from left to right on $\ell$.
%Sometimes we use indices to refer to points of $P$. For example, point $i$ refers to $p_i$.
%Besides, let $w_i$ denote its weight, assume that each $w_i$ is positive (otherwise one could always include $p_i$ in the solution).
For any point $p\in P$, we use $w(p)$ to denote its weight. We assume that $w(p)>0$ for each $p\in P$ since otherwise one could always include $p$ in the solution.

We sort all disks of $S$ by their centers from left to right; let $s_1, s_2,
\cdots , s_m$ be the sorted list. For each disk $s_j\in S$, let $l_j$ and $r_j$ denote its leftmost and rightmost points on $\ell$, respectively. Note that $l_j$ is the leftmost point of $s_j$ and $r_j$ is the
rightmost point of $s_j$. More specifically, $l_j$ (resp., $r_j$) is the only leftmost (resp., rightmost) point of $s_j$ in the 1D, $L_1$, and $L_2$ cases. For each of exposition, we make a general position assumption that
no two points of $\{l_i,r_i\ |\ 1\leq i\leq m\}$ are coincident.
%Let $\calC$ denote the set of centers of all disks. For each disk $s_i$, let $l(s_i)$ (resp., $r(s_i)$) represent the leftmost (resp., rightmost) of $P$ which hits $s_i$.
For $1 \leq j_1 \leq j_2 \leq m$, let $S[j_1, j_2]$ denote the subset of disks $s_j\in S$ for all $j\in [j_1,j_2]$.
%$\calC[i, j]$ denote the subset ${c_i, c_{i+1},\cdots, c_j}$.

We often talk about the relative positions of two geometric objects $O_1$ and $O_2$ (e.g., two points, or a point and a line). We say that $O_1$ is to the {\em left} of $O_2$ if $x(p) \leq x(p')$ holds for any point $p \in O_1$ and any point $p' \in O_2$, and {\em strictly left} means $x(p) < x(p')$. Similarly, we can define {\em right, above, below,} etc.

%We use the term \textit{optimal solution} $P_{opt}$ to refer to a subset of $P$ used in an optimal solution, and use \textit{optimal point value} $w_{opt}$ to refer to the total sum of the weights of the points in $P_{opt}$.

\subsection{Non-Containment subset}
\label{sec:fifo}

We observe that to solve the problem it suffices to consider only a subset of $S$ with certain property, called the {\em Non-Containment subset}, defined as follows.
We say that a disk of $S$ is {\em redundant} if it contains another disk of $S$. The Non-Containment subset, denoted by $\widehat{S}$, is defined as the subset of $S$ excluding all redundant disks. We have the following observation on $\widehat{S}$, which is called the {\em Non-Containment property}.

\begin{observation}{\em (Non-Containment Property)}\label{obser:FIFO}
For any two disks $s_i, s_j \in \widehat{S}$, $x(l_i) < x(l_j)$ if and only if $x(r_i) < x(r_j)$.
\end{observation}

It is not difficult to see that it suffices to work on $\widehat{S}$ instead of $S$. Indeed, suppose $P_{opt}$ is an optimal solution for $\widehat{S}$. Then, for any disk $s\in S\setminus \widehat{S}$, there must be a disk $s'\in \widehat{S}$ such that $s$ contains $s'$. Hence, any point of $P_{opt}$ hitting $s'$ must hit $s$ as well.

We can easily compute $\widehat{S}$ in $O(m \log m)$ time in any metric. Indeed, because all disks of $S$ are centered at $\ell$, a disk $s_k$ contains another disk $s_j$ if and only the segment $s_k\cap \ell$ contains the segment $s_j\cap \ell$. Hence, it suffices to identify all redundant segments from $\{s_j\cap \ell\ | \ s_j\in S \}$. This can be easily done in $O(m\log m)$ time, e.g., by sweeping the endpoints of disks on $\ell$; we omit the details.

In what follows, to simplify the notation, we assume that $S=\widehat{S}$, i.e., $S$ does not have any redundant disk. As such, $S$ has the Non-Containment property in Observation~\ref{obser:FIFO}. As will be seen later, the Non-Containment property is very helpful in designing algorithms.

\section{Dual transformation and the 1D, unit-disk, and $L_1$ problems}
\label{sec:dualtransformation}

By making use of the Non-Containment property of $S$, we propose a {\em dual transformation} that can reduce our hitting set problem on $S$ and $P$ to an instance of the 1D dual coverage problem. More specifically, we will construct a set $S^*$ of points and a set $P^*$ of weighted segments on the $x$-axis such that an optimal solution for the coverage problem on $S^*$ and $P^*$ corresponds to an optimal solution for our original hitting set problem. We refer to it as the {\em 1D dual coverage problem}. To differentiate from the original hitting set problem on $P$ and $S$, we refer to the points of $S^*$ as {\em dual points} and the segments of $P^*$ as {\em dual segments}.

As will be seen later, $|S^*|=m$, but $|P^*|$ varies depending on the specific problem. Specifically, $|P^*|\leq n$ for the 1D, unit-disk, and $L_1$ cases, $|P^*|=O(n+m)$ for the $L_{\infty}$ case, and $|P^*|=O(m+\kappa)$ for the $L_2$ case. In what follows, we present the details of the dual transformation by defining $S^*$ and $P^*$.

%We first define the dual point set $S^*$.
For each disk $s_j\in S$, we define a dual point $s_j^*$ on the $x$-axis with $x$-coordinate equal to $j$. Define $S^*$ as the set of all $m$ points $s_1^*,s_2^*,\ldots,s_m^*$. As such, $|S^*|=m$.

We next define the set $P^*$ of dual segments. For each point $p_i\in P$, let $I_i$ be the set of indices of the disks of $S$ that are hit by $p_i$. We partition the indices of $I_i$ into maximal intervals of consecutive indices and let $\calI_i$ be the set of all these intervals. By definition, for each interval $[j_1,j_2]\in \calI_i$, $p_i$ hits all disks $s_j$ with $j_1\leq j\leq j_2$ but does not hit either $s_{j_1-1}$ or $s_{j_2+1}$; we define a dual segment on the $x$-axis whose left (resp., right) endpoint has $x$-coordinate equal to $j_1$ (resp., $j_2$) and whose weight is equal to $w(p_i)$ (for convenience, we sometimes also use the interval $[j_1,j_2]$ to represent the dual segment and refer to dual segments as intervals). We say that the dual segment is {\em defined} or {\em generated} by $p_i$. Let $P^*_i$ be the set of dual segments defined by the intervals of $\calI_i$. We define $P^*=\bigcup_{i=1}^nP^*_i$. The following observation follows the definition of dual segments.

\begin{observation}\label{obser:20}
$p_i$ hits a disk $s_j$ if and only if a dual segment of $P^*_i$ covers the dual point $s_j^*$.
\end{observation}

Suppose we have an optimal solution $P^*_{opt}$ for the 1D dual coverage problem on $P^*$ and $S^*$, we obtain an optimal solution $P_{opt}$ for the original hitting set problem on $P$ and $S$ as follow: for each segment of $P^*_{opt}$, if it is from $P^*_i$ for some $i$, then we include $p_i$ into $P_{opt}$.
%We argue the correctness below.

Clearly, $|S^*|=m$. We will prove later in this section that $|P^*_i|\leq 1$ for all $1\leq i\leq n$ in the 1D problem, the unit-disk case, and the $L_1$ metric, and thus $|P^*|\leq n$ for all these cases. Since $|P^*_i|\leq 1$ for all $1\leq i\leq n$, in light of Observation~\ref{obser:20}, $P_{opt}$ constructed above is an optimal solution of the original hitting set problem. Therefore, one can solve the original hitting set problem for the above cases with the following three main steps: (1) Compute $S^*$ and $P^*$; (2) apply the algorithm for the 1D dual coverage problem in~\cite{ref:PedersenAl22} to compute $P^*_{opt}$, which takes $O((|S^*|+|P^*|)\log (|S^*|+|P^*|))$ time~\cite{ref:PedersenAl22}; (3) derive $P_{opt}$ from $P^*_{opt}$. For the first step, computing $S^*$ is straightforward. For $P^*$, we will show later that for all above three cases (1D, unit-disk, $L_1$), $P^*$ can be computed in $O((n+m)\log (n+m))$ time. As $|S^*|=m$ and $|P^*|\leq n$, the second step can be done in $O((n+m)\log (n+m))$ time~\cite{ref:PedersenAl22}. As such, the hitting set problem of the above three cases can be solved in $O((n+m)\log (n+m))$ time.

For the $L_{\infty}$ metric, we will prove in Section~\ref{sec:infinity} that $|P^*|=O(n+m)$ but each $P^*_i$ may have multiple segments. If $P^*_i$ has multiple segments, a potential issue is the following: If two segments of $P^*_i$ are in $P^*_{opt}$, then the weights of both segments will be counted in the optimal solution value (i.e., the total weight of all segments of $P^*_{opt}$), which corresponds to counting the weight of $p_i$ twice in $P_{opt}$. To resolve the issue, we prove in Section~\ref{sec:infinity} that even if $|P^*_i|\geq 2$, at most one dual segment of $P^*_i$ will appear in any optimal solution $P^*_{opt}$. As such, $P_{opt}$ constructed above is an optimal solution for the original hitting set problem. Besides proving the upper bound $|P^*|=O(n+m)$, another challenge of the $L_{\infty}$ problem is to compute $P^*$ efficiently, for which we propose an $O((n+m)\log (n+m))$ time algorithm. Consequently, the $L_{\infty}$ hitting set problem can be solved in  $O((n+m)\log (n+m))$ time.
%There are two challenges for the $L_2$ problem: (1) prove the upper bound $|P^*|=O(m+n)$; (2) compute $P^*$ efficiently.

For the $L_2$ metric, we will show in Section~\ref{sec:l2} that $|P^*|=O(m+\kappa)$. Like the $L_{\infty}$ case, each $P^*_i$ may have multiple segments but we can also prove that $P^*_i$ can contribute at most one segment to any optimal solution $P^*_{opt}$. Hence, $P_{opt}$ constructed above is an optimal solution for the original hitting set problem.
We present an algorithm that can compute $P^*$ in $O((n+m)\log (n+m)+\kappa\log m)$ time. As such, the $L_2$ hitting set problem can be solved in $O((m+n)\log(m+n)+\kappa \log m)$ time. Alternatively, a straightforward approach can prove $|P^*|=O(nm)$ and compute $P^*$ in $O(nm)$ time; hence, we can also solve the problem in $O(nm\log (n+m))$ time.

%\begin{figure}[t]
%\begin{minipage}[t]{\textwidth}
%\begin{center}
%\includegraphics[height=1.0in]{newFig/dual.png}
%\caption{\footnotesize Illustrating a dual transformation example. $P$ = $\{p_i,p_j\}$, $S$ = $\{s_i,s_j,s_k\}$. The dual object set $\Bar{\mathcal{H}}(p_i) = \{h_{i,k}\}$,$\Bar{\mathcal{H}}(p_j) = \{h_{i,i},h_{k,k}\}$.(Note: $h_{i,i} = [c_i,c_i]$) The dual point set $\mathcal{C} = \{c_i,c_j,c_k\}$.
%}
%\label{fig:dual}
%\end{center}
%\end{minipage}
%\vspace{-0.15in}
%\end{figure}

In the rest of this section, following the above framework, we will solve the 1D problem, the unit-disk case, and the $L_1$ case in Sections~\ref{sec:1Dproblem}, \ref{sec:unitDist}, and \ref{sec:L_1}, respectively.

\subsection{The 1D problem}
\label{sec:1Dproblem}
In the 1D problem, all points of $P$ are on $\ell$ and each disk $s_i\in S$ is a line segment on $\ell$, and thus $l_i$ and $r_i$ are the left and right endpoints of $s_i$, respectively.
We follow the above dual transformation and have the following lemma.

\begin{lemma}\label{lem:10}
In the 1D problem, $|P_i^*|\leq 1$ for all $1\leq i\leq n$. In addition, $P_i^*$ for all $1\leq i\leq n$ can be computed in $O((n+m)\log (n+m))$ time.
\end{lemma}
\begin{proof}
Consider a point $p_i\in P$. If $p_i$ does not hit any disk, then $|P_i^*|=0$. Otherwise, since $S$ has the Non-Containment property, the indices of the segments of $S$ hit by $p_i$ must be consecutive. Hence, $|P_i^*|=1$. This proves the first part of the lemma.

To compute $P_i^*$ for all $1\leq i\leq n$, we use a straightforward sweeping algorithm. We sweep a point $q$ on $\ell$ from left to right. During the sweeping, we store in $Q$ all disks hit by $q$ sorted by their indices. When $q$ encounters the left endpoint of a disk $s_j$, we add $s_j$ to the rear of $Q$. When $q$ encounters the right endpoint of a disk $s_j$, $s_j$ must be at the front of $Q$ due to the Non-Containment property of $S$ and we remove $s_j$ from $Q$. When $q$ encounters a point $p_i$, we report $P_i^*=\{[j_1,j_2]\}$, where $j_1$ (resp., $j_2$) is the index of the front (resp., rear) disk of $Q$. After the endpoints of all disks of $S$ and the points of $P$ are sorted on $\ell$ in $O((n+m)\log (n+m))$ time, the above sweeping algorithm can be implemented in $O(n+m)$ time. \qed
\end{proof}

In light of Lemma~\ref{lem:10}, using the dual transformation, the 1D hitting set problem can be solved in $O((n+m)\log (n+m))$ time. The result is summarized in the following theorem, whose proof also provides a simple dynamic programming algorithm that solves the problem directly.

\begin{theorem}
The line-constrained 1D hitting set problem can be solved in $O((n+m)\log (n+m))$ time.
\end{theorem}
\begin{proof}
In addition to the above method using the dual transformation and applying the 1D dual coverage algorithm~\cite{ref:PedersenAl22}, we present below a simple dynamic programming algorithm that solves the problem directly; the runtime of the algorithm is also $O((n+m)\log (n+m))$.

For each point $p_i\in P$, let $a_i$ refer to the largest index of the disk in $S$ whose right endpoint is strictly left of $p_i$, i.e., $a_i=\arg\max_{1\leq j\leq m}\{s_j \in S': x(r_j)<x(p_i)\}$. Due to the Non-Containment property of $S$, the indices $a_i$ for all $i=1,2,\ldots,n$ can be obtained in $O(n+m)$ time after we sort all points of $P$ along with the endpoints of all segments of $S$.
%More specifically, we sweep a point $q$ from left to right on $L$ after sorting all points in $P\cup \{l_k,r_k\},k \in [1,m]$. During the sweep, we maintain a counter $i$ to mark the number of the disks in $S'$ have swept to the strictly left of $q$, $i=0$ at initial.
%When $q$ encounters the right endpoint of a segment $s_j\in S'$, we increment $i$ by one. When $q$ encounters a point $p_j \in P$, we set $f(j)=i$. No event happens when $q$ encounters a left endpoint of the disks.

For each $j\in [1,m]$, define $W(j)$ to be the minimum total weight of any subset of points of $P$ that hit all disks of $S[1,j]$. Our goal is thus to compute $W(m)$. For convenience, we set $W(0)=0$.
For each point $p_i\in P$, we define its {\em cost} as $cost(i)=w(p_i)+W(a_i)$.
As such, $W(j)$ is equal to the minimum $cost(i)$ among all points $p_i\in P$ that hit $s_j$. This is the recursive relation of our dynamic programming algorithm.

We sweep a point $q$ on $\ell$ from left to right. During the sweeping, we maintain the subset $P'$ of all points of $P$ that are to the left of $q$ and the cost values for all points of $P'$ as well as the values $W(j)$ for all disks $s_j$ whose right endpoints are to the left of $q$. An event happens when $q$ encounters a point of $P$ or the right endpoint of a segment of $S$.
%To guide the sweeping, we sort all endpoints of the segments of $S'$ along with the points of $P$.
If $q$ encounters a point $p_i\in P$, we set $cost(i)=w(p_i)+W(a_i)$ and insert $p_i$ into $P'$. If $q$ encounters the right endpoint of a segment $s_j$, then among the points of $P'$ that hit $s_j$, we find the one with minimum cost and set $W(j)$ to the cost value of the point.
If we store the points of $P'$ by an augmented balanced binary search tree with their $x$-coordinates as keys and each node storing the minimum cost of all leaves in the subtree rooted at the node, then processing each event can be done in $O(\log n)$ time.

As such, the sweeping takes $O((n+m)\log n)$ time, after sorting the points of $P$ and all segment endpoints in $O((n+m)\log(n+m))$ time. \qed
\end{proof}

\subsection{The unit-disk case}
\label{sec:unitDist}

In the unit-disk case, all disks of $S$ have the same radius. We follow the dual transformation and have the following lemma.

\begin{lemma}\label{lem:20}
In the unit-disk case, $|P_i^*|\leq 1$ for any $1\leq i\leq n$. In addition, $P_i^*$ for all $1\leq i\leq n$ can be computed in $O((n+m)\log (n+m))$ time.
\end{lemma}
\begin{proof}
Consider a point $p_i\in P$. Observe that $p_i$ hits a disk $s_j$ if and only
if the segment $D(p_i)\cap \ell$ covers the center of $s_j$, where $D(p_i)$ is the
	unit disk centered at $p_i$. By definition, the indices of the disks whose
centers are covered by the segment $D(p_i)\cap \ell$ must be consecutive. Hence, $|P_i^*|\leq 1$ must hold.

To compute $P_i^*$, it suffices to determine the disks whose
centers are covered by $D(p_i)\cap \ell$. This can be easily done in $O((n+m)\log
(n+m))$ time for all $p_i\in P$
(e.g., first sort all disk centers and then do binary search on the
sorted list with the two endpoints of $D(p_i)\cap \ell$ for each $p_i\in P$).
\qed
\end{proof}

In light of Lemma~\ref{lem:20}, using the dual transformation, the unit-disk case can be solved in $O((n+m)\log (n+m))$ time.

\begin{theorem}
The line-constrained unit-disk hitting set problem can be solved in $O((n+m)\log (n+m))$ time.
\end{theorem}

\subsection{The $L_1$ metric}
\label{sec:L_1}

In the $L_1$ metric, each disk of $S$ is a diamond, whose boundary is comprised
of four edges of slopes 1 or -1, but the diamonds of $S$ may have different
radii.
We follow the dual transformation and have the following lemma.

\begin{lemma}\label{lem:30}
In the $L_1$ metric, $|P_i^*|\leq 1$ for any $1\leq i\leq n$. In addition, $P_i^*$ for all $1\leq i\leq n$ can be computed in $O((n+m)\log (n+m))$ time.
\end{lemma}
\begin{proof}
Assume to the contrary that $|P_i^*|>1$. Let $[j_1,j_2]$ and
$[j_3,j_4]$ be two consecutive intervals. Hence, $j_2+1 \leq j_3-1$, and $p_i$ is in the common
intersection $C$ of the four disks $s_{j_1}$, $s_{j_2}$, $s_{j_3}$, and $s_{j_4}$,
while $p_i$ does not hit $s_j$ for any $j\in [j_2+1,j_3-1]$.
Note that $C$ is also a diamond with its leftmost and rightmost points on
$\ell$.
Further, due to the Non-Containment property of $S$, the leftmost point of $C$ is $l_{j_4}$
and the rightmost endpoint is $r_{j_1}$ (e.g., see Fig.~\ref{fig:fourdisks}).

\begin{figure}[t]
\begin{minipage}[t]{\textwidth}
\begin{center}
\includegraphics[height=1.1in]{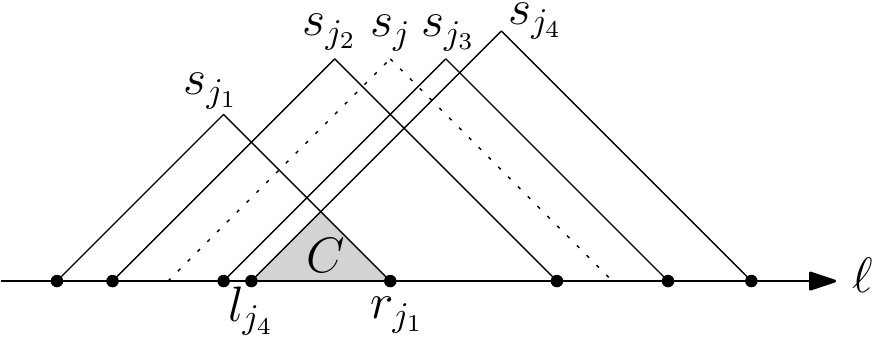}
\caption{\footnotesize Illustrating the disks $s_{j_1}$, $s_{j_2}$, $s_{j_3}$, $s_{j_4}$, and $s_j$ for some $j\in [j_2+1,j_3-1]$; only the portions of the disks above $\ell$ are shown.}
\label{fig:fourdisks}
\end{center}
\end{minipage}
\vspace{-0.15in}
\end{figure}

On the other hand, consider any $j\in [j_2+1,j_3-1]$. Since $j_1<j<j_4$, due to the Non-Containment property of $S$,
$x(l_j)\leq x(l_{j_4})$ and $x(r_{j_1})\leq x(r_j)$, implying that $C\subseteq s_j$
since both $C$ and $s_j$ are diamonds (e.g., see Fig.~\ref{fig:fourdisks}). As $p_i\in C$, $p_i$ must hit
$s_j$. But this incurs contradiction since $p_i$ does not hit $s_j$.

This proves that $|P_i^*|\leq 1$ for any $1\leq i\leq n$.

In the following, we describe an algorithm to compute $P_i^*$ for all $1\leq i\leq n$.

We sweep a vertical line $\ell'$ in the plane from left to right. During the sweeping we
maintain two subsets $S_L$ and $S_R$ of $S$: $S_L$ (resp., $S_R$) consists of
all disks of $S$ whose upper left (resp., right) edges intersecting
$\ell'$; disks of $S_L$ (resp., $S_R$) are stored in a binary
search tree $T_L$ (resp., $T_R$) sorted by the $y$-coordinates of the
intersections between $\ell'$ and the upper left (resp., right) edges of the disks
of $S_L$ (resp., $S_R$).
An event happens if $\ell'$ encounters a point of $P$, the left endpoint,
the right endpoint, or the center of a disk $s_j$.

If $\ell'$ encounters the left endpoint of a disk $s_j$, we insert $s_j$ into $T_L$. If
$\ell'$ encounters the center of a disk $s_j$, we remove $s_j$ from $T_L$ and insert
it into $T_R$. If $\ell'$ encounters the right endpoint of a disk $s_j$, we remove
$s_j$ from $T_R$. If $\ell'$ encounters a point $p_i\in P$, we compute the only interval
$[j_1,j_2]$ of $P_i^*$ as follows.

\begin{figure}[h]
\begin{minipage}[t]{\textwidth}
\begin{center}
\includegraphics[height=1.2in]{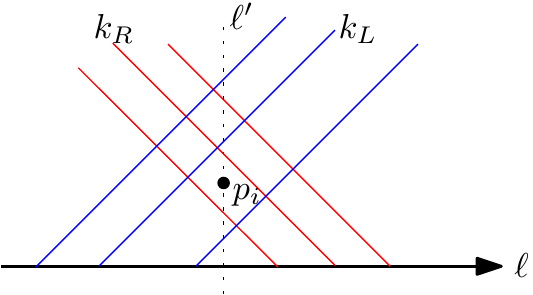}
\caption{\footnotesize Illustrating the processing of the event at $p_i$: The red segments are the upper right edges of disks in $T_R$ and the blue segments are upper left edges of disks in $T_L$.}
\label{fig:sweepdisk}
\end{center}
\end{minipage}
\vspace{-0.15in}
\end{figure}

Using $T_R$, we find the disk of $T_R$ whose
upper right edge is the lowest but above $p_i$;
let $k_R$ be the index of the disk (e.g., see Fig.~\ref{fig:sweepdisk}). Similarly, we find the disk
of $T_L$ whose upper left edge is the lowest but above $p_i$; let $k_L$ be the
index of the disk. Both $k_R$ and $k_L$ can be found in $O(\log m)$ time.

Assuming that both $k_R$ and $k_L$ are well defined, we claim that $j_1=k_R$ and
$j_2=k_L$. Indeed, for any disk $s_j\in T_R$ that is below $s_{k_R}$, $p_i$ does
not hit $s_j$ and $j<k_R$ due to the Non-Containment property of $S$. On the other hand,
for any disk $s_j\in T_R$ that is above $s_{k_R}$, $p_i$
hits $s_j$ and $j>k_R$ due to the Non-Containment property of $S$. Similarly, for any
disk $s_j\in T_L$ that is below $s_{k_L}$, $p_i$ does
not hit $s_j$ and $j>k_L$, and for any disk $s_j\in T_L$ that is above $s_{k_L}$, $p_i$
hits $s_j$ and $j<k_L$. Note that the indices of disks in $T_L$ are larger than
those in $S_R$ due to the Non-Containment property of $S$. Also note that disks not in $T_L$ or $T_R$ cannot be hit by $p_i$. As such, $j_1=k_R$ and $j_2=k_L$ must hold.

The above argument assumes that both $k_R$ and $k_L$ are well defined. If
neither $k_R$ nor $k_L$ exists, then $P_i^*=\emptyset$. If $k_R$ exists while
$k_L$ does not, then $j_1=k_R$ and $j_2$ is the index of the highest disk of
$T_R$. If $k_L$ exists while $k_R$ does not, then $j_2=k_L$ and $j_1$ is the
highest disk of $T_L$. The proof is similar to the above and we omit the
details.

It is not difficult to see that the above sweeping algorithm can be implemented in $O((n+m)\log (n+m))$
time.
%once the leftmost and the rightmost points as well as the centers of all disks are sorted in $O((n+m)\log(n+m))$ time. The lemma thus follows.
\qed
\end{proof}

In light of Lemma~\ref{lem:30}, using the dual transformation, the $L_1$ case can be solved in $O((n+m)\log (n+m))$ time.

\begin{theorem}
The line-constrained $L_1$ hitting set problem can be solved in $O((n+m)\log (n+m))$ time.
\end{theorem}

%\section{The $L_{\infty}$ and $L_2$ metrics}
%\label{sec:lInftyandl2}
%
%In this section, we describe our algorithms for the $L_{\infty}$ and $L_2$ cases.
%Then, we devise the weighted algorithms for $L_{\infty}$ and $L_2$ metrics in
%Sections~\ref{sec:weightedlinfinity} and ~\ref{sec:weightedl2}, respectively.
%In Section~\ref{sec:UnweightedL2Linfinity}, we present simpler algorithms  for both cases for their unweighted cases, where all points of $P$ have the same weight.

\section{The $L_{\infty}$ metric}
\label{sec:infinity}
In this section, following the dual transformation, we present an $O((m+n)\log(m+n))$ time algorithm for $L_{\infty}$ case.

In the $L_{\infty}$ metric, each disk is a square whose edges are axis-parallel.
%For any two disks $s$ and $s'$, since its two upper edges are horizontal, we say that $s$ is {\em higher than}) $s'$ if the upper edge of $s$ is higher than that of $s'$.
For a disk $s_j\in S$ and a point $p_i\in P$, we say that $p$ is {\em vertically above} $s_j$ if $p_i$ is outside $s_j$ and $x(l_j)\leq x(p_i)\leq x(r_j)$.

In the $L_{\infty}$ metric, using the dual transformation, it is easy to come up with an example in which $|P_i^*|\geq 2$. Observe that $|P_i^*|\leq \lceil m/2\rceil$ as the indices of $S$ can be partitioned into at most $\lceil m/2\rceil$ disjoint maximal intervals. Despite $|P_i^*|\geq 2$, the following critical lemma shows that each $P_i^*$ can contribute at most one segment to any optimal solution of the 1D dual coverage problem on $P^*$ and $S^*$.

\begin{lemma}\label{lem:40}
In the $L_{\infty}$ metric, for any optimal solution $P^*_{opt}$ of the 1D dual coverage problem on $P^*$ and $S^*$, $P^*_{opt}$ contains at most one segment from $P_i^*$ for any $1\leq i\leq n$.
\end{lemma}
\begin{proof}
Assume to the contrary that $P^*_{opt}$ contains more than one segment from $P_i^*$. Among all segments of $P^*_{opt}\cap P_i^*$, we choose two consecutive segments (recall that no two segments of $P_i^*$ are overlapped); we let $[j_1,j_2]$ and $[j_3,j_4]$ denote these two segments, respectively, with $j_2+1\leq j_3-1$ from $P_i^*$.
Then all disks in $S[j_1,j_2]\cup S[j_3,j_4]$ are hit by $p_i$, while $s_j$ is not hit by $p_i$ for any
$j\in [j_2+1,j_3-1]$.

We claim that $p_i$ is vertically above $s_{j}$ for any $j\in [j_2+1,j_3-1]$. To see this, since
$j_2<j<j_3$, due to the Non-Containment property of $S$, $x(l_{j})\leq x(l_{j_3})$ and
$x(r_{j})\geq x(r_{j_2})$. As $p_i$ hits both $s_{j_2}$ and $s_{j_3}$, we have
$x(l_{j_3})\leq x(p_i)\leq x(r_{j_2})$. As such, we obtain that
$x(l_{j})\leq x(p_i)\leq x(r_{j})$. Since $p_i$ does not hit $s_{j}$, $p_i$ must
be vertically above $s_{j}$.

Let $\ell_{p_i}$ be the vertical line through $p_i$. The above claim implies that the upper edges of all disks of $S[j_2+1,j_3-1]$ intersect $\ell_{p_i}$.
Among all disks of $S[j_2+1,j_3-1]$, let $s_{j_0}$ be the one whose upper edge is the lowest.

Since $P^*_{opt}$ is an optimal solution to the 1D dual coverage problem, one
dual segment $[j_5,j_6]\in P^*_{opt}$ defined by some point $p_{i'}$ with $i\neq i'$ must
cover the dual point $s^*_{j_0}$, i.e., $p_{i'}$ hits all disks $s_j$ with $j\in [j_5,j_6]$ and $j_0\in [j_5,j_6]$. In particular, $p_{i'}$ hits $s_{j_0}$.
In what follows, we prove that $[j_5,j_6]$ must contain either $[j_1,j_2]$ or $[j_3,j_4]$.
Depending on whether $x(p_i')\leq x(p_i)$, there are two cases.

\begin{itemize}
  \item
  If $x(p_{i'})\leq x(p_i)$, we prove below that $p_{i'}$ hits all disks $S[j_1,j_0]$. Recall that $p_{i'}$ hits $s_{j_0}$. Hence, it suffices to prove that $p_{i'}$ hits $s_j$ for any $j\in [j_1,j_0-1]$.

  Consider any $j\in [j_1,j_0-1]$. We claim that the upper edge of $s_j$ must be higher than that of $s_{j_0}$. Indeed, if $j\in [j_2+1,j_0-1]$, then the claim is obviously true by the definition of $j_0$. Otherwise, $j\in [j_1,j_2]$ and thus $p_i$ hits $s_j$. Hence, $p_i$ is lower than the upper edge of $s_j$. As $p_i$ is vertically above $s_{j_0}$, we obtain that the upper edge of $s_j$ must be higher than that of $s_{j_0}$. The claim thus follows.

  Since $j<j_0$, due to the Non-Containment property of $S$, we have $x(l_j)\leq x(l_{j_0})$. Recall that $s_{j_0}$ intersects $\ell_{p_i}$. Since $p_i$ hits $s_j$, $s_{j}$ also intersects $\ell_{p_i}$. As such, since the upper edge of $s_j$ is higher than that of $s_{j_0}$, the portion of $s_{j_0}$ to the left of $\ell_{p_i}$ is a subset of the portion of $s_j$ to the left of $\ell_{p_i}$ (e.g., see Fig.~\ref{fig:contain}). As $p_{i'}$ hits $s_{j_0}$ and $x(p_{i'})\leq x(p_i)$, $p_{i'}$ is inside the portion of $s_{j_0}$ to the left of $\ell_{p_i}$. Therefore, $p_{i'}$ is inside the portion of $s_j$ to the left of $\ell_{p_i}$. Hence, $p_{i'}$ hits $s_j$.

\begin{figure}[h]
\begin{minipage}[t]{\textwidth}
\begin{center}
\includegraphics[height=1.0in]{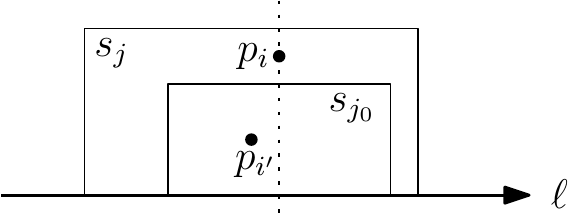}
\caption{\footnotesize Illustrating the proof of Lemma~\ref{lem:40}.
}
\label{fig:contain}
\end{center}
\end{minipage}
\vspace{-0.15in}
\end{figure}

  This proves that $p_{i'}$ hits all disks of $S[j_1,j_0]$. As $j_0\in [j_5,j_6]$, $[j_1,j_0]$ must be contained in $[j_5,j_6]$ since $[j_5,j_6]$ is a maximal interval of indices of disks hit by $p_{i'}$. Since $[j_1,j_2]\subseteq [j_1,j_0]$, we obtain that $[j_5,j_6]$ must contain $[j_1,j_2]$.
  \item
  If $x(p_{i'})> x(p_i)$, then by a symmetric analysis to the above, we can show that $[j_5,j_6]$ must contain $[j_3,j_4]$.
\end{itemize}

The above proves that $[j_5,j_6]$ contains either $[j_1,j_2]$ or $[j_3,j_4]$. Without loss of generality, we assume that $[j_5,j_6]$ contains $[j_1,j_2]$.
As $[j_5,j_6]$ is in $P^*_{opt}$, if we remove $[j_1,j_2]$ from $P^*_{opt}$, the rest of  the intervals of $P^*_{opt}$ still form a coverage for all dual points of $S^*$, which contradicts with that $P^*_{opt}$ is an optimal coverage.

The lemma thus follows.
\qed
\end{proof}

The above lemma implies that an optimal solution to the 1D dual coverage problem on $P^*$ and $S^*$ still corresponds to an optimal solution of the original hitting set problem on $P$ and $S$. As such, it remains to compute the set $P^*$ of dual segments. In what follows, we first prove an upper bound for $|P^*|$.

\subsection{Upper bound for $|P^*|$}
%As $|P^*_i|\geq 2$ is possible, the next issue is to give an upper bound for $|P^*|$.
As $|P^*_i|\leq \lceil m/2\rceil$, an obvious upper bound for $|P^*|$ is $O(mn)$. In the following, we reduce it to $O(m+n)$.

Our first observation is that if the same dual segment of $P^*$ is defined by more than one point of $P$, then we only need to keep the one whose weight is minimum. In this way, all segments of $P^*$ are distinct (i.e., $P^*$ is not a multi-set).

We sort all points of $P$ from top to bottom as $q_1,q_2,\ldots, q_n$. For ease of exposition, we assume that no point of $P$ has the same $y$-coordinate as the upper edge of any disk of $S$.
For each
$2\leq i\leq n$, let $S_i$ denote the subset of disks whose upper edges are
between $q_{i-1}$ and $q_{i}$. Let $S_1$
denote the subset of disks whose upper edges are above $q_1$.
For each $1\leq i\leq n$, let $m_i=|S_i|$.

We partition the indices of disks of $S_1$ into a set $\calI_1$ of maximal
intervals. Clearly, $|\calI_1|\leq m_1$. The next lemma shows that other than the dual segments corresponding to the intervals in $\calI_1$, $q_1$ can generate at most two dual segments in $P^*$.

\begin{lemma}\label{lem:50}
The number of dual segments of $P^*\setminus \calI_1$ defined by $q_1$ is at most
$2$.
\end{lemma}
\begin{proof}
Assume to the contrary that $q_1$ defines three intervals $[j_1,j_1']$,
$[j_2,j_2']$, and $[j_3,j_3']$ in $P^*\setminus \calI_1$, with $j_1'<j_2$ and
$j_2'<j_3$. By definition, $\calI_1$ must have
an interval, denoted by $I_k$, that strictly contains $[j_k,j_k']$ (i.e., $[j_k,j_k']\subset I_k$), for each $1\leq k\leq 3$.
Then, $I_2$ must contain an index $j$ that is
not in $[j_1,j_1']\cup [j_2,j_2']\cup [j_3,j_3']$ with $j_1'<j<j_3$
(e.g., see Fig.~\ref{fig:dualsegments}).
As such, $q_1$ does not hit $s_j$. Also, since $j\in I_2$, $s_j$ is in $S_1$.

\begin{figure}[h]
\begin{minipage}[t]{\textwidth}
\begin{center}
\includegraphics[height=0.9in]{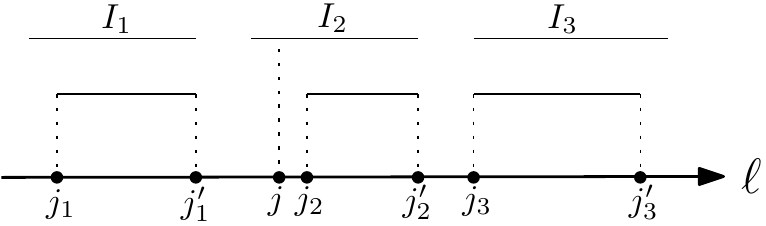}
\caption{\footnotesize Illustrating a schematic view of the intervals $[j_k,j_k']$ and $I_k$ for $1\leq k\leq 3$.}
\label{fig:dualsegments}
\end{center}
\end{minipage}
\vspace{-0.15in}
\end{figure}

Since $j_1'<j<j_3$, due to the Non-Containment property of $S$, $x(l_j)\leq x(l_{j_3})$ and
$x(r_{j_1'})\leq x(r_j)$. As $q_1$ hits both $s_{j_1'}$ and $s_{j_3}$, we have
$x(l_{j_3})\leq x(p_1)\leq x(r_{j_1'})$. Hence, we obtain $x(l_j)\leq x(q_1)\leq x(r_j)$. Since
$q_1$ does not hit $s_j$, the upper edge of $s_j$ must be below $q_1$. But this implies that $s_j$ is
not in $S_1$, which incurs contradiction.
\qed
\end{proof}

Now we consider the disks of $S_2$ and the dual segments defined by $q_2$. For each
disk $s_j$ of $S_2$, we update the intervals of $\calI_1$ by adding the
index $j$, as follows. Note that by definition, intervals of $\calI_1$ are pairwise disjoint and no interval contains $j$.

\begin{enumerate}
\item
If neither $j+1$ nor $j-1$ is in any interval of
$\calI_1$, then we add $[j,j]$ as a new interval to $\calI_1$.

\item
If $j+1$ is
contained in an interval $I\in \calI_1$ while $j-1$ is not, then $j+1$ must be the left endpoint of
$I$. In this case, we add $j$ to $I$ to obtain a new interval $I'$ (which has $j$ as its left endpoint) and add $I'$ to $\calI_1$; but we still keep
$I$ in $\calI_1$.

\item
Symmetrically, if $j-1$ is contained in an interval $I\in \calI_1$ while $j+1$ is not, then we add $j$ to $I$ to obtain a new interval $I'$ and add $I'$ to $\calI_1$; we still keep $I$ in $\calI_1$.

\item
If both $j+1$ and $j-1$ are contained in intervals of $\calI_1$, then they must
be contained in two intervals, respectively; we merge these two
intervals into a new interval by padding $j$ in between and adding the new interval to $\calI_1$. We
still keep the two original intervals in $\calI_1$.
\end{enumerate}

Let $\calI_1'$ denote the updated set $\calI_1$ after the above operation. Clearly, $|\calI_1'|\leq |\calI_1|+1$.

We process all disks $s_j\in S_2$ as above; let $\calI_2$ be
the resulting set of intervals. It holds that $|\calI_2|\leq
|\calI_1|+|S_2|\leq m_1+m_2$. Also observe that for any interval $I$ of indices of
disks of $S_1\cup S_2$ such that $I$ is not in $\calI_2$, $\calI_2$ must have an
interval $I'$ such that $I\subset I'$ (i.e., $I\subseteq I'$ but $I\neq I'$).
Using this property, by exactly the same
analysis as Lemma~\ref{lem:50}, we can show that other than the intervals in
$\calI_2$, $q_2$ can generate at most two intervals in $P^*$.
Since $\calI_1\subseteq \calI_2$,
combining Lemma~\ref{lem:50}, we obtain that other than the intervals of
$\calI_2$, the number of intervals of $P^*$ generated by $q_1$ and $q_2$ is at
most $4$.

We process disks of $S_i$ and point $q_i$ in the same way as above for all
$i=3,4,\dots,n$. Following the same argument, we can show that for each $i$, we obtain an interval set $\calI_i$ with $\calI_{i-1}\subseteq \calI_i$ and $|\calI_i|\leq \sum_{k=1}^i m_k$, and other than the intervals of
$\calI_i$, the number of intervals of $P^*$ generated by $\{q_1,q_2,\ldots,q_i\}$ is at
most $2i$. In particular, $|\calI_n|\leq \sum_{k=1}^n m_k\leq m$, and other than the intervals of
$\calI_n$, the number of intervals of $P^*$ generated by $P=\{q_1,q_2,\ldots,q_n\}$ is at
most $2n$. We thus achieve the following conclusion.

\begin{lemma}\label{lem:60}
In the $L_{\infty}$ metric, $|P^*|\leq 2n+m$.
\end{lemma}

\subsection{Computing $P^*$}

Using Lemma~\ref{lem:60}, we next present an algorithm that computes $P^*$ in $O((n+m)\log (n+m))$ time.

For each segment $I\in P^*$, let $w(I)$ denote its weight.
We say that a segment $I$ of $P^*$ is {\em redundant} if there is another segment $I'$
such that $I\subset I'$ and $w(I)\geq w(I')$. Clearly, any redundant segment of
$P^*$ cannot be used in any optimal solution for the 1D dual coverage problem on
$S^*$ and $P^*$. A segment of $P^*$ is {\em non-redundant} if it is not redundant.

In the following algorithm, we will compute a subset $P^*_0$ of $P^*$ such that
segments of $P^*\setminus P^*_0$ are all redundant (i.e., the segments of $P^*$ that are not computed by the algorithm are all redundant and thus are useless). We will show that each
segment reported by the algorithm belongs to $P^*$ and thus the total number of
reported segments is at most $2n+m$ by Lemma~\ref{lem:60}. We will show that the
algorithm spends $O(\log(n+m))$ time reporting one segment and each segment is reported only once; this guarantees the $O((n+m)\log (n+m))$ upper bound of the runtime of the algorithm.

For each disk $s_j\in S$, we use $y(s_j)$ to denote the $y$-coordinate of the upper edge of $s_j$.

Our algorithm has $m$ iterations. In the $j$-th iteration, it computes all segments in $P^*_j$, where $P^*_j$ is the set of all non-redundant segments of $P^*$ whose starting indices are $j$, although it is possible that some redundant segments with starting index $j$ may also be computed. Points of $P$ defining these segments must be inside $s_j$; let $P_j$ denote the set of points of $P$ inside $s_j$. We partition $P_j$ into two subsets (e.g., see Fig.~\ref{fig:types}): $P_j^1$ consists of points of $P_j$ to the left of $r_{j-1}$ and $P_j^2$ consists of points of $P_j$ to the right of $r_{j-1}$. We will compute dual segments of $P^*_j$
defined by $P_j^1$ and $P_j^2$ separately; one reason for doing so is that when computing dual segments defined by a point of $P_j^1$, we need to additionally check whether this point also hits $s_{j-1}$ (if yes,
such a dual segment does not exist in $P^*$ and thus will not be reported).
In the following, we first describe the algorithm for $P_j^1$ since the algorithm for $P_j^2$ is basically the same but simpler. Note that our algorithm does not need to explicitly determine the points of $P_j^1$ or $P_j^2$; rather we will build some data structures that can implicitly determine them during certain queries.

\begin{figure}[t]
\begin{minipage}[t]{\textwidth}
\begin{center}
\includegraphics[height=1.5in]{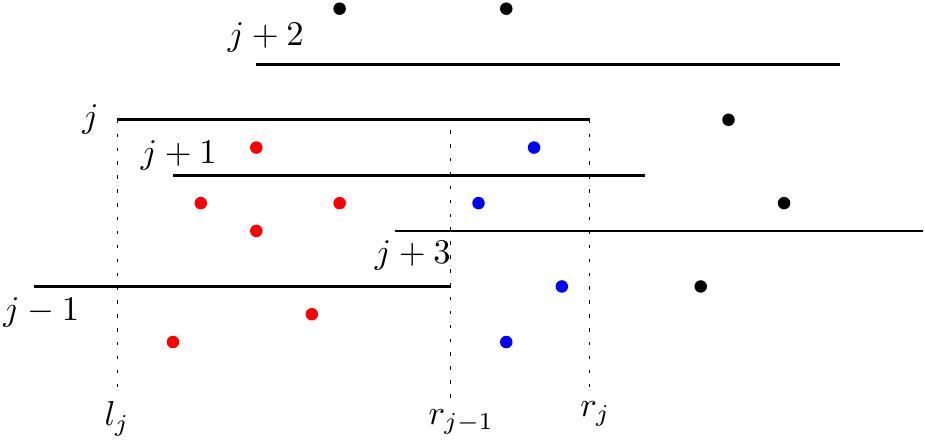}
\caption{\footnotesize Illustrating $P_j^1$ (the red points) and $P_j^2$ (the blue points). Only the upper edges of disks are shown. The numbers are the indices of disks.
}
\label{fig:types}
\end{center}
\end{minipage}
\vspace{-0.15in}
\end{figure}

If the upper edge of $s_{j-1}$ is higher than that of $s_j$, then all points of $P_j^1$ are in $s_{j-1}$ and thus no point of $P_j^1$ defines any dual segment of $P^*$ starting from $j$. Indeed, assume to the contrary that a point $p_i\in P_j^1$ defines such a dual segment $[j,j']$. Then, since $p_i$ is in $s_{j-1}$, $[j,j']$ cannot be a maximal interval of indices of disks hit by $p_i$ and thus cannot be a dual segment defined by $p_i$. In what follows, we assume that the upper edge of $s_{j-1}$ is lower than that of $s_j$. In this case, it suffices to only consider points of $P_j^1$ above $s_{j-1}$ since points below the upper edge of $s_{j-1}$ (and thus are inside $s_{j-1}$) cannot define any dual segments due to the same reason as above. Nevertheless, our algorithm does not need to explicitly determine these points.
%Let $P'$ denote the subset of points of $P_j^1$ above $s_{j-1}$. We will compute all non-redundant dual intervals defined by points of $P'$.

We start with performing the following {\em rightward segment dragging query}: Drag the vertical segment $x(l_j)\times [y(s_{j-1}),y(s_j)]$ rightwards until a point $p\in P$ and return $p$ (e.g., see Fig.~\ref{fig:rightdrag}). Such a segment dragging query can be answered in $O(\log n)$ time after $O(n\log n)$ time preprocessing on $P$ (e.g., using Chazelle's result~\cite{ref:ChazelleAn88} one can build a data structure of $O(n)$ space in $O(n\log n)$ time such that each query can be answered in $O(\log n)$ time; alternatively, if one is satisfied with an $O(n\log n)$ space data structure, then an easier solution is to use fractional cascading~\cite{ref:ChazelleFr86} and one can build a data structure in $O(n\log n)$ time and space with $O(\log n)$ query time).
If the query does not return any point or if the query returns a point $p$ with $x(p)>x(r_{j-1})$, then $P_j^1$ does not have any point above $s_{j-1}$ and we are done with the algorithm for $P_j^1$. Otherwise, suppose the query returns a point $p$ with $x(p)\leq x(r_{j-1})$; we proceed as follows.

\begin{figure}[h]
\begin{minipage}[t]{\textwidth}
\begin{center}
\includegraphics[height=1.3in]{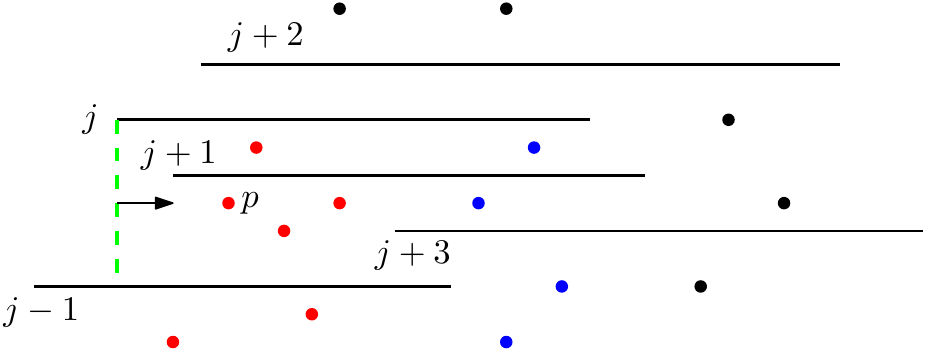}
\caption{\footnotesize Illustrating the rightward segment dragging query: the green dashed segment is the dragged segment $x(l_j)\times [y(s_{j-1}),y(s_j)]$.
}
\label{fig:rightdrag}
\end{center}
\end{minipage}
\vspace{-0.15in}
\end{figure}

We perform the following {\em max-range query} on $p$: Compute the largest index $k$ such that all disks in $S[j,k]$ are hit by $p$ (e.g., in Fig.~\ref{fig:rightdrag}, $k=j+2$). We will show later in Lemma~\ref{lem:maxrange} that after $O(m\log m)$ time and $O(m)$ space processing, each such query can be answered in $O(\log m)$ time. Such an index $k$ must exist as $s_j$ is hit by $p$. Observe that $[j,k]$ is a dual segment in $P^*$ defined by $p$. However, the weight of $[j,k]$ may not be equal to $w(p)$, because it is possible that a point with smaller weight also defines $[j,k]$. Our next step is to determine the minimum-weight point that defines $[j,k]$.

We perform a {\em range-minima query} on $[j,k]$: Find the lowest disk among all disks in $S[j,k]$ (e.g., in Fig.~\ref{fig:rightdrag}, $s_{j+1}$ is the answer to the query). This can be easily done in $O(\log m)$ time with $O(m)$ space and $O(m\log m)$ time preprocessing. Indeed, we can build a binary search tree on the upper edges of all disks of $S$ with their $y$-coordinates as keys and have each node storing the lowest disk among all leaves in the subtree rooted at the node. A better but more complicated solution is to build a range-minima data structure on the $y$-coordinates of the upper edges of all disks in $O(m)$ time and each query can be answered in $O(1)$ time~\cite{ref:BenderTh00,ref:HarelFa84}. However, the above binary search tree solution is sufficient for our purpose.
Let $y^*$ be the $y$-coordinate of the upper edge of the disk returned by the query.

We next perform the following {\em downward min-weight point query} for the horizontal segment
$[x(l_k),x(r_{j-1})]\times y^*$: Find the minimum weight point of $P$ below the
segment (e.g., see Fig.~\ref{fig:minweight}). We will show later in Lemma~\ref{lem:minweightquery}
that after $O(n\log n)$ time and space
preprocessing, each query can be answered in $O(\log n)$ time.
Let $p'$ be the point returned by the query. If $p'=p$, then we report $[j,k]$
as a dual segment with weight equal to $w(p)$. Otherwise, if $p'$ is inside $s_{j-1}$ or $s_{k+1}$, then $[j,k]$ is a redundant dual segment (because a dual segment defined by $p'$ strictly contains $[j,k]$ and $w(p')\leq w(p)$) and thus we do not need to report it.
In any case, we proceed as follows.

\begin{figure}[h]
\begin{minipage}[t]{\textwidth}
\begin{center}
\includegraphics[height=1.5in]{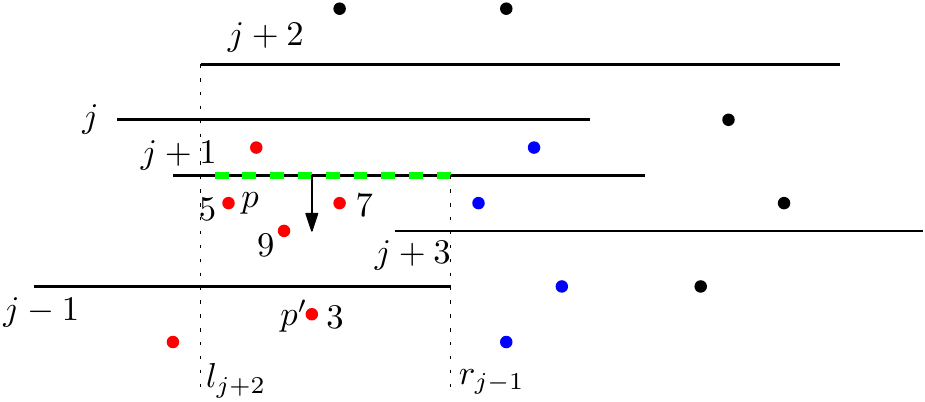}
\caption{\footnotesize Illustrating the downward min-weight point query (with $k=j+2$): the green dashed segment is the dragged segment $[x(l_k),x(r_{j-1})]\times y^*$. The numbers besides the points are their weights. The answer to the query is $p'$, whose weight is $3$.
}
\label{fig:minweight}
\end{center}
\end{minipage}
\vspace{-0.15in}
\end{figure}

The above basically determines that $[j,k]$ is a dual segment in $P^*$. Next, we proceed to determine those dual segments $[j,k']$ with $k'>k$. If such a dual segment exists, the interval $[j,k']$ must contain index $k+1$. Hence, we next consider $s_{k+1}$. If $y(s_{k+1})>y(s_{j-1})$, then let $y'=\min\{y^*,y(s_{k+1})\}$; we perform a rightward segment dragging query with the vertical segment
$x(l_{k+1})\times[y(s_{j-1}),y']$ (e.g., see Fig.~\ref{fig:rightdrag2}) and then repeat the above algorithm.
If $y(s_{k+1})\leq y(s_{j-1})$, then points of $P_j^1$ above $s_{j-1}$ are also above $s_{k+1}$ and
thus no point of $P_j^1$ can generate any dual segment $[j,k']$ with $k'>k$ and
thus we are done with the algorithm on $P_j^1$.

\begin{figure}[h]
\begin{minipage}[t]{\textwidth}
\begin{center}
\includegraphics[height=1.2in]{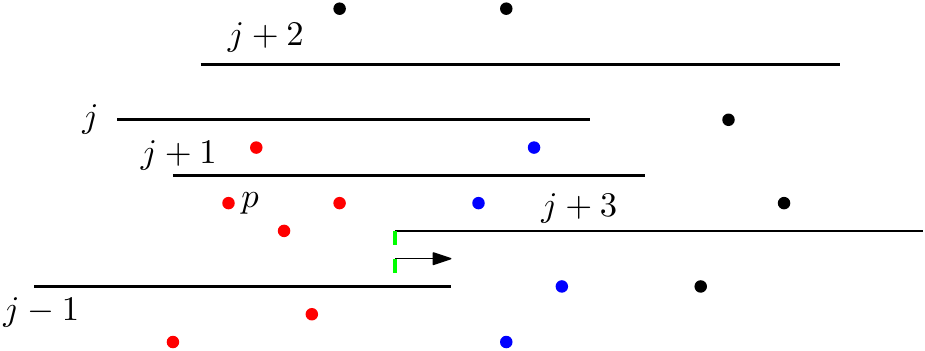}
\caption{\footnotesize Illustrating the rightwards segment dragging query: the green dashed segment is the dragged segment $x(l_{k+1})\times[y(s_{j-1}),y']$.
}
\label{fig:rightdrag2}
\end{center}
\end{minipage}
\vspace{-0.15in}
\end{figure}

For the time analysis, we charge the time of the above five queries to the
interval $[j,k]$, which is in $P^*$. Note that $[j,k]$ will not be charged again in the future because future queries in the $j$-th iteration will be charged to $[j,k']$ for some $k'>k$ and future queries in the $j'$-th iteration for any $j'>j$ will be charged to $[j',k'']$. As such, each dual segment of $P^*$ will be charged $O(1)$ times during the entire algorithm. As each query takes $O(\log (n+m))$
time, the total time of all queries in the entire algorithm is $O(|P^*|\log
(n+m))$, which is $((n+m)\log (n+m))$ by Lemma~\ref{lem:60}.

\begin{lemma}\label{lem:maxrange}
With $O(m\log m)$ time and $O(m)$ space preprocessing on $S$, each max-range query can be answered in $O(\log m)$ time.
\end{lemma}
\begin{proof}
We build a complete binary search tree $T$ with $m$ leaves storing the disks of $S$ in their index order. For each node $v\in T$, we store a value $y_v$ that is equal to the minimum $y(s_j)$ for all disks $s_j$ stored in the leaves of the subtree rooted at $v$.
In addition, we use an array $A$ to store all disks sorted by their indices. This finishes our preprocessing, which takes $O(m\log m)$ time and $O(m)$ space.

Given a query point $p$ and a disk index $j$ with $p$ hitting $s_j$, the max-range query asks for the largest index $k$ such that all disks of $S[j,k]$ are hit by $p$. Our query algorithm has two main steps. In the first step, we use the tree $T$ to find in $O(\log m)$ time the largest index $k'$ such that $y(s_{t})\geq y(p)$ for all $t\in [j,k']$; the details of the algorithm will be described later. In the second step, using the array $A$, we find the largest index $k''$ such that $x(l_{k''})\leq x(p)$. As the disks in $A$ are sorted by their indices, due to the Non-Containment property, the disks $s_j$ of $A$ are also sorted by the values $x(l_j)$. Hence, $k''$ can be found in $O(\log m)$ time by binary search on $A$. As $p$ hits $s_j$, we have $x(l_j)\leq x(p)$ and thus $j\leq k''$. After having $k'$ and $k''$, we return $k=\{k',k''\}$ as the answer to the max-range query. In the following, we prove the correctness: $k$ thus defined is the largest index such that all disks of $S[j,k]$ are hit by $p$. Depending on whether $k'\leq k''$, there are two cases.

\begin{enumerate}
  \item If $k'\leq k''$, then $k=k'$. By the definition of $k'$, $y(s_{k'+1})< y(p)$ and thus $p$ does not hit $s_{k'+1}$. Hence, if suffices to prove that $p$ hits $s_{t}$ for all $t\in [j,k']$. Indeed, since $t\leq k'\leq k''$, by the definition of $k''$, we have $x(l_t)\leq x(p)$. On the other hand, since $p$ hits $s_j$, we have $x(p)\leq x(r_j)$. Since $j\leq t$, by the Non-Containment property of $S$, $x(r_j)\leq x(r_t)$. Therefore, we obtain $x(p)\leq x(r_t)$. Finally, as $t\leq k'$, by the definition of $k'$, $y(p)\leq y(s_t)$.

      In summary, we have $x(l_t)\leq x(p)\leq x(r_t)$ and $y(p)\leq y(s_t)$. Therefore, $p$ hits $s_t$. This proves that $k=k'$ is the largest index such that all disks of $S[j,k]$ are hit by $p$.

  \item
      If $k'> k''$, then $k=k''$. By the definition of $k''$, $x(l_{k''+1})> x(p)$ and thus $p$ does not hit $s_{k''+1}$. Hence, if suffices to prove that $p$ hits $s_{t}$ for all $t\in [j,k'']$. Indeed, since $t\leq k''$, by the definition of $k''$, we have $x(l_t)\leq x(p)$. On the other hand, since $p$ hits $s_j$, we have $x(p)\leq x(r_j)$. Since $j\leq t$, by the Non-Containment property of $S$, $x(r_j)\leq x(r_t)$. Therefore, we obtain $x(p)\leq x(r_t)$. Finally, as $t\leq k''<k'$, by the definition of $k'$, $y(p)\leq y(s_t)$.

      In summary, we have $x(l_t)\leq x(p)\leq x(r_t)$ and $y(p)\leq y(s_t)$. Therefore, $p$ hits $s_t$. This proves that $k=k''$ is the largest index such that all disks of $S[j,k]$ are hit by $p$.
\end{enumerate}

It remains to describe the algorithm for computing $k'$ using $T$. The algorithm has two phases. Starting from the leaf storing disk $s_j$, for each node $v$, we process it as follows. Let $u$ be the parent of $v$. If $v$ is the right child of $u$, then we proceed on $u$ recursively by setting $v=u$. If $v$ is the left child of $u$, then let $w$ be the right child of $u$. If $y_w\geq y(p)$, then we proceed on $u$ recursively by setting $v=u$. Otherwise, the first phase of the algorithm is over and the second phase starts from $v=w$ in a top-down manner as follows. Let $u$ and $w$ be the left and right children of $v$ recursively. If $y_u\geq y(p)$, then we proceed on $w$ recursively by setting $v=w$; otherwise we proceed on $u$ recursively by setting $v=u$. When we reach a leaf $v$, which stores a disk $s_t$, if $y(s_t)\geq y(p)$, then we return $k'=t$; otherwise we return $k'=t-1$. Clearly, the algorithm runs in $O(\log m)$ time.

The lemma thus follows. \qed
\end{proof}

\begin{lemma}\label{lem:minweightquery}
With $O(n\log n)$ time and space preprocessing on $P$, each downward min-weight point
query can be answered in $O(\log n)$ time.
\end{lemma}
\begin{proof}
Recall that the downward min-weight point query is to compute the minimum weight point of
$P$ below a query horizontal segment.

We built a complete binary search tree $T$ whose leaves store points of $P$ from left to right.
For each node $v\in T$, let $P_v$ denote the subset of points of $P$ in the
leaves of the subtree rooted at $v$. We compute a subset $P_v'\subseteq P_v$
with the following property: (1) If we sort all points of $P_v'$ in the order of
decreasing $y$-coordinate, then the weights of the points are sorted in
increasing order; (2) for any point $p\in P_v\setminus P_v'$, $P_v'$ must have a point $p'$ below $p$ with $w(p')\leq w(p)$. We compute $P_v'$ for all $v\in T$ in
a bottom-up manner as follows. Initially, let $P_v'=P_v$ for all leaves $v\in
T$. Consider an internal node $v$. We assume that both $P_u'$ and $P_w'$ are
computed already, where $u$ and $w$ are the two children of $v$. We also assume
that points of both $P_u'$ and $P_w'$ are sorted by $y$-coordinate. The subset $P_v'$ is computed
by merging $P_u'$ and $P_w'$ as follows.

We scan the two sorted lists of $P_u'$ and $P_w'$ in decreasing $y$-coordinate
order, in the same way as merge sort. Suppose we are comparing two points $p_u\in
P_u'$ and $p_w\in P_w'$ and the higher one will be placed at the end of an
already sorted list $L$ (assume that the lowest point of $L$ is higher than both $p_u$ and $p_w$; initially $L=\emptyset$). Suppose $p_u$ is higher than $p_w$. In the normal merge
sort, one would just place $p_u$ at the end of $L$. Here we do the following.
Let $p$ be the lowest point of $L$. If $w(p_u)>w(p)$, then add $p_u$ to the end
of $L$. Otherwise, we remove $p$ from $L$ (we say that $p$ is {\em pruned}), and then we keep pruning the next lowest point of $L$ until its weight is smaller than $w(p_u)$ and finally we place $p_u$ at the end of $L$.
Clearly, the time for computing
$P_v'$ is bounded by $O(|P_u'|+|P_w'|)$.

In this way, we can compute $P_v'$ for all nodes $v\in T$ in $O(n\log n)$ time
and space.
Next, we construct a fractional cascading data structure~\cite{ref:ChazelleFr86}
on the sorted lists of $P_v'$ of all nodes $v\in T$, which can be done in time
linear to the total size of all lists, which is $O(n\log n)$.
This finishes the preprocessing, which takes $O(n\log n)$ time and space.

Given a query horizontal segment $B=[x_1,x_2]\times y$, our goal is to find the
minimum weight point among all points of $P$ below $B$. Using the standard
approach, we can find in $O(\log n)$ time a set $V$ of $O(\log n)$ nodes such
that the union $\bigcup_{v\in V}P_v$ is exactly the subset of points of $P$
whose $x$-coordinates are in $[x_1,x_2]$ and parents of nodes of $V$ are on two
paths of $T$ from the root to two nodes. For each node $v\in V$, we wish to
find the highest point $p_v$ of $P_v'$ below $B$. Due to the above Property (2) of $P_v'$, $p_v$ must be the
minimum weight point below $B$ among all points of $P_v$. We can compute $p_v$ for all $v\in V$
in $O(\log n)$ time using the fractional cascading data
structure~\cite{ref:ChazelleFr86}, after which we return the highest $p_v$ among
all $v\in V$ as the answer to the query. The total time of the query algorithm is $O(\log n)$.

This proves the lemma.
\qed
\end{proof}

This finishes the description of the algorithm for $P_j^1$. The algorithm for
$P_j^2$ is similar with the following minor changes. First, when doing each rightward
segment dragging query, the lower endpoint of the query vertical segment is
at $-\infty$ instead of $y(s_{j-1})$. Second, when the downward min-weight point query
returns a point, we do not have to check whether it is in $s_{j-1}$ anymore.
The rest of the algorithm is the same.
In this way, all non-redundant intervals of $P^*$ starting at index $j$ can be
computed. As analyzed above, the runtime of the entire algorithm is bounded by
$O((n+m)\log (n+m))$.

As such, using the dual transformation, the $L_{\infty}$ case can be solved in
$O((n+m)\log (n+m))$ time.

\begin{theorem}
The line-constrained $L_{\infty}$ hitting set problem can be solved in $O((n+m)\log (n+m))$ time.
\end{theorem}

%8015812169

\section{The $L_2$ case}
\label{sec:l2}
In this section, following the dual transformation, we solve $L_{2}$ hitting set problem.

Recall that we have made a general position assumption that no point of $P$ is on the
boundary of a disk of $S$. In the $L_2$ metric, $l_j$ (resp., $r_j$) is the only leftmost (resp., rightmost) point of disk $s_j$. For a disk $s_j\in S$ and a point $p_i\in P$, we say that $p_i$ is {\em vertically above} $s_j$ if $p_i$ is outside $s_j$ and $x(l_j)\leq x(p_i)\leq x(r_j)$. For any disk $s_j$, we use $\partial s_j$ to denote the portion of its boundary above $\ell$, which is a half-circle. Note that $\partial s_j$ and $\partial s_k$ have at most one intersection, for any two disks $s_j$ and $s_k$.

%Consider a disk $s_j$ in $S$. If there exist two disks $s_{j_1}$ and $s_{j_2}$ in $S$ with $j_1<j<j_2$ such that $s_j\subseteq s_{j_1}\cup s_{j_2}$, then we call $s_j$ is a {\em gap disk} (e.g., see Fig.~\ref{fig:Liniftygap}).
%
%
%\begin{figure}[h]
%\begin{minipage}[t]{\textwidth}
%\begin{center}
%\includegraphics[height=1.0in]{newFig/LinftyGap.png}
%\caption{\footnotesize Illustration for the gap-disk $s_k$. There are two disks vertically above it, and $i\leq k\leq j$. $\Bar{\mathcal{H}}(p_i) = \{h_{i,i},h_{j,j}\}$ instead of $\{h_{i,j}\}$. $h_{i,i}$ is a segment that have same point $c_i$ for left and right endpoint.
%}
%\label{fig:Liniftygap}
%\end{center}
%\end{minipage}
%\vspace{-0.15in}
%\end{figure}

As in the $L_{\infty}$ case, it is possible that $|P_i^*|\geq 2$;
$|P_i^*|\leq \lceil m/2\rceil$ also holds. We have the following lemma, whose proof
follows the scheme of Lemma~\ref{lem:40} although the details are not exactly the same.

\begin{lemma}\label{lemma:80}
In the $L_2$ metric, for any optimal solution $P^*_{opt}$ to the 1D dual coverage problem on $P^*$ and $S^*$, $P^*_{opt}$ contains at most one dual segment from $P_i^*$ for any $1\leq i\leq n$.
\end{lemma}
\begin{proof}
Assume to the contrary that $P^*_{opt}$ contains more than one segment from $P_i^*$. Among all segments of $P^*_{opt}\cap P_i^*$, we choose two consecutive segments $[j_1,j_2]$ and $[j_3,j_4]$; thus $j_2+1\leq j_3-1$.
Then all disks in $S[j_1,j_2]\cup S[j_3,j_4]$ are hit by $p_i$, while $s_j$ is not hit by $p_i$ for any
$j\in [j_2+1,j_3-1]$.

Due to the Non-Containment property of $S$,
following the argument in Lemma~\ref{lem:40}, $p_i$ is vertically above $s_{j}$ for any $j\in [j_2+1,j_3-1]$.
Among all disks of $S[j_2+1,j_3-1]$, let $s_{j_0}$ be the one whose boundary has the lowest intersection with $\ell_{p_i}$, where $\ell_{p_i}$ is the vertical line through $p_i$.

%To see this, since $j_2<j_0<j_3$, due to the Non-Containment property of $S$, $x(l_{j_0})<x(l_{j_3})$ and $x(r_{j_0})>x(r_{j_2})$. As $p_i$ hits both $s_{j_2}$ and $s_{j_3}$, we have $x(l_{j_3})<x(p_i)<x(r_{j_2})$. As such, we obtain that $x(l_{j_0})<x(p_i)<x(r_{j_0})$. Since $p_i$ does not hit $s_{j_0}$, $p_i$ must be vertically above $s_{j_0}$.
Since $P^*_{opt}$ is an optimal solution, a dual
segment $[j_5,j_6]\in P^*$ defined by some point $p_{i'}$ with $i\neq i'$ must
cover the dual point $s^*_{j_0}$, i.e., $p_{i'}$ hits all disks $s_j$ with $j\in [j_5,j_6]$ and $j_0\in [j_5,j_6]$. In particular, $p_{i'}$ hits $s_{j_0}$.
In what follows, we prove that $[j_5,j_6]$ must contain either $[j_1,j_2]$ or $[j_3,j_4]$.
Depending on whether $x(p_{i'})\leq x(p_i)$, there are two cases.

\begin{figure}[h]
\begin{minipage}[t]{\textwidth}
\begin{center}
\includegraphics[height=1.3in]{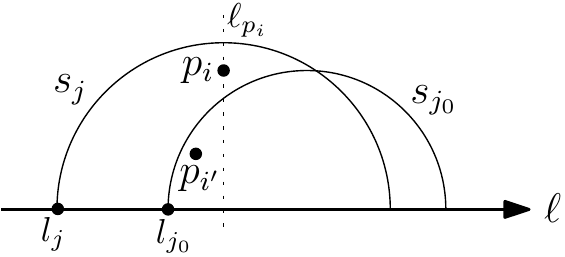}
\caption{\footnotesize Illustrating the proof of Lemma~\ref{lemma:80}.
}
\label{fig:cross}
\end{center}
\end{minipage}
\vspace{-0.15in}
\end{figure}

\begin{itemize}
  \item
    If $x(p_{i'})\leq x(p_i)$, we prove below that $p_{i'}$ hits all disks $S[j_1,j_0]$. Recall that $p_{i'}$ hits $s_{j_0}$. Hence, it suffices to prove that $p_{i'}$ hits $s_j$ for any $j\in [j_1,j_0-1]$. Consider any $j\in [j_1,j_0-1]$.

    We claim that the intersection of $\partial s_j$ with $\ell_{p_i}$ must be higher than that of $\partial s_{j_0}$. Indeed, if $j\in [j_2+1,j_0-1]$, then this follows the definition of $j_0$. If $j\leq j_2$, then $p_i$ must hit $s_j$. Hence, $p_i$ is below $\partial s_j$. Because $p_i$ is vertically above $s_{j_0}$, the claim follows.

    Since $j<j_0$, due to the Non-Containment property, $x(l_j)<x(l_{j_0})$. Since $\partial s_j$ and $\partial s_{j_0}$ cross each other at most once, the above claim implies that $s_j$ and $s_{j_0}$ do not cross each other on the left side of $\ell_{p_i}$ (e.g., see Fig.~\ref{fig:cross}). As $p_{i'}$ is to the left of $\ell_{p_i}$ and $p_{i'}$ is inside $s_{j_0}$, the above further implies that $p_{i'}$ is inside $s_j$ as well (e.g., see Fig.~\ref{fig:cross}).

      This proves that $p_{i'}$ hits all disks of $S[j_1,j_0]$. As $j_0\in [j_5,j_6]$, $[j_1,j_0]$ must be contained in $[j_5,j_6]$ since $[j_5,j_6]$ is a maximal interval of indices of disks hit by $p_{i'}$. Since $[j_1,j_2]\subseteq [j_1,j_0]$, we obtain that $[j_5,j_6]$ must contain $[j_1,j_2]$.
  \item
  If $x(p_{i'})> x(p_i)$, then by a symmetric analysis to the above, we can show that $[j_5,j_6]$ must contain $[j_3,j_4]$.
\end{itemize}

The above proves that $[j_5,j_6]$ contains either $[j_1,j_2]$ or $[j_3,j_4]$. Without loss of generality, we assume that $[j_5,j_6]$ contains $[j_1,j_2]$.
As $[j_5,j_6]$ is in $P^*_{opt}$, if we remove $[j_1,j_2]$ from $P^*_{opt}$, the rest of  the intervals of $P^*_{opt}$ still form a coverage for all dual points of $S^*$, which contradicts with that $P^*_{opt}$ is an optimal coverage.

The lemma thus follows.
\qed
\end{proof}

As in the $L_{\infty}$ case, the above lemma implies that it suffices to find an
optimal solution to the 1D dual coverage problem on $P^*$ and $S^*$.

\subsection{Upper bound for $|P^*|$}
\label{sec:boundL2}

As $|P_i^*|\leq \lceil m/2\rceil$, an obvious upper bound for $|P^*|$ is
$O(mn)$. In this section, with some observations, we show that
$|P^*|=O(m+\kappa)$, where $\kappa$ is the number of pairs of disks of $S$ that
intersect.

Let $H$ denote the upper half-plane bounded by $\ell$.
Consider two disks $s_{j}$ and $s_{k}$ whose boundaries intersect, say, at a point $v$. The boundaries of $s_{j}$ and $s_{j'}$ partition $H$ into four regions. One region is inside both disks; another region is outside both disks. Each of the remaining two regions is contained in exactly one of the disks; we call them the {\em wedges} of $v$ (resp., $s_{j}$ and $s_{j'}$). One wedge has $v$ as its rightmost point and we call it the {\em left wedge}; the other wedge has $v$ as its leftmost point and we call it the {\em right wedge} (e.g., see Fig.~\ref{fig:wedge}).

\begin{figure}[h]
\begin{minipage}[t]{\textwidth}
\begin{center}
\includegraphics[height=1.0in]{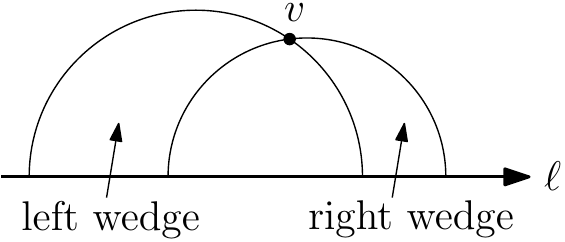}
\caption{\footnotesize Illustrating left and right wedges of two disks that intersect.}
\label{fig:wedge}
\end{center}
\end{minipage}
\vspace{-0.15in}
\end{figure}

Let $\calA$ be the arrangement of the boundaries of all disks of $S$ in the half-plane $H$. $\calA$ has a single face that is outside all disks; for convenience, we remove it from $\calA$. Observe that points of $P$ located in the same face of $\calA$ define the same subset of dual segments of $P^*$. Hence it suffices to consider the dual segments defined by all faces of $\calA$.

Due to that all disks of $S$ are centered on $\ell$ as well as the Non-Containment property of $S$, we discuss some properties of the faces of $\calA$. For convenience, we consider $\ell$ as the boundary of a disk with an infinite radius and with its center below $\ell$ (thus $H$ is the region outside the disk); let $s'$ denote the disk and let $S'=S\cup \{s'\}$. Consider a face $f\in \calA$. Each edge $e$ of $f$ is a circular arc on the boundary of a disk $s_e$ of $S'$, such that $f$ is either inside or outside the disk $s_e$. More specifically, if $u$ and $v$ are leftmost and rightmost vertices of $f$, respectively, then $u$ and $v$ partition the boundary of $f$ into an upper chain and a lower chain (both chains are $x$-monotone). For each edge $e$ in the upper (resp., lower) chain, $f$ is inside (resp., outside) $s_e$.
%We call the edges in the upper chain (resp., lower) the {\em upper edges} of $f$ and the edges in the lower chain the {\em lower edges} of $f$.
As the boundaries of every two disks of $S$ cross each other at most once,
the boundary of each disk contributes at most one edge of $f$.
Consider the leftmost vertex $u$ of $f$, which is incident two edges of $f$, one edge $e_a$ on the upper chain and the other edge $e_b$ on the lower chain. If $e_b$ is on $\ell$, then $u$ is actually the leftmost point of the disk whose boundary contains $e_a$; in this case, we call $f$ an {\em initial face}  (e.g., see Fig.~\ref{fig:initialface}).
If $f$ is not an initial face, we say that $f$ is a {\em non-initial face}; in this case,
$u$ must be the rightmost vertex of another face $f'$ such that $f'$ is in the left wedge of $u$ while $f$ is in the right wedge of $u$, and we call $f'$ the {\em opposite face} of $f$ (e.g., see Fig.~\ref{fig:noninitial}).

\begin{figure}[h]
\begin{minipage}[t]{\textwidth}
\begin{center}
\includegraphics[height=1.0in]{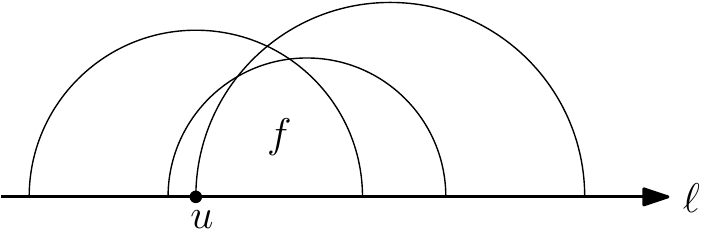}
\caption{\footnotesize Illustrating an initial face $f$ with leftmost vertex $u$.
}
\label{fig:initialface}
\end{center}
\end{minipage}
\vspace{-0.15in}
\end{figure}

\begin{figure}[h]
\begin{minipage}[t]{\textwidth}
\begin{center}
\includegraphics[height=1.0in]{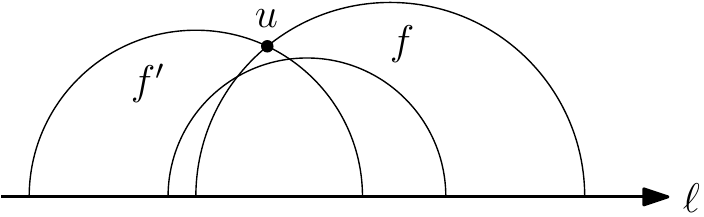}
\caption{\footnotesize Illustrating a non-initial face $f$ with leftmost vertex $u$ and its opposite face $f'$.}
\label{fig:noninitial}
\end{center}
\end{minipage}
\vspace{-0.15in}
\end{figure}

To help the analysis, we introduce a directed graph $G$, defined as follows. The faces of $\calA$ form the node set of $G$. There is an edge from a node $f'$ to another node $f$ if the face $f$ is a non-initial face and $f'$ is the opposite face of $f$ (i.e., the rightmost vertex of $f'$ is the leftmost vertex of $f$; e.g., in Fig.~\ref{fig:noninitial}, there is a directed edge from $f'$ to $f$). Since each face of $\calA$ has only one leftmost vertex and only one rightmost vertex, each node $G$ has at most one incoming edge and at most one outgoing edge. Observe that each initial face does not have an incoming edge while each non-initial face must have an incoming edge. As such, $G$ is actually composed of a set of paths, each of which has an initial face as the first node.

For each face $f\in \calA$, we use $P^*(f)$ to denote the subset of the dual segments of $P^*$ generated by $f$ (i.e, generated by any point in $f$). Our goal is to obtain an upper bound for $|\bigcup_{f\in \calA}P^*(f)|$, which will be an upper bound for $|P^*|$ as $P^*\subseteq\bigcup_{f\in \calA}P^*(f)$.
The following lemma proves that each initial face can only generate one dual segment.

\begin{lemma}\label{lem:90}
For each initial face $f$, $|P^*(f)| = 1$.
\end{lemma}
\begin{proof}
Let $u$ be the leftmost vertex of $f$. By the definition of initial faces, $u$ is the leftmost point $l_k$ of a disk $s_k$ and $f\subseteq s_k$. In the following, we show that indices of all disks of $S$ containing $l_k$ must form an interval $[k',k]$ for some $k'\leq k$, which will prove the lemma.

Indeed, for any disk $s_j$ with $j>k$, due to the Non-Containment property of $S$, $x(l_k)<x(l_{j})$ and thus $s_j$ cannot contain $l_k$. On the other hand, suppose $s_j$ contains $l_k$ for some $j<k$. Then, for any $j'$ with $j<j'<k$, we claim that $s_{j'}$ contains $l_k$. It suffices to show that $x(l_{j'})<x(l_k)<x(r_{j'})$. Due to the Non-Containment property of $S$, we have $x(l_{j'})<x(l_k)$. Also, since $s_j$ contains $l_k$, we have $x(l_j)<x(l_k)<x(r_j)$. Due to the Non-Containment property, since $j<j'$, we have $x(r_j)<x(r_{j'})$. As such, we obtain $x(l_k)<x(r_{j'})$. The claim thus follows, which leads to the lemma.
\qed
\end{proof}

The next lemma shows that for any two adjacent faces $f'$ and $f$ in any path of $G$, the symmetric difference between $P^*(f')$ and $P^*(f)$ is of constant size.

\begin{lemma}\label{lem:100}
For any two adjacent faces $f'$ and $f$ in any path of $G$, $|P^*(f)\setminus P^*(f')|\leq 3$ and $|P^*(f')\setminus P^*(f)|\leq 3$.
\end{lemma}
\begin{proof}
As $f'$ and $f$ are adjacent in a path of $G$, without loss of generality, we assume that there is a directed edge from $f'$ to $f$. By definition, $f'$ and $f$ share a vertex $u$ that is the rightmost vertex of $f'$ and also the leftmost vertex of $f$ (e.g., see Fig.~\ref{fig:adjdisks}).

\begin{figure}[h]
\begin{minipage}[t]{\textwidth}
\begin{center}
\includegraphics[height=1.0in]{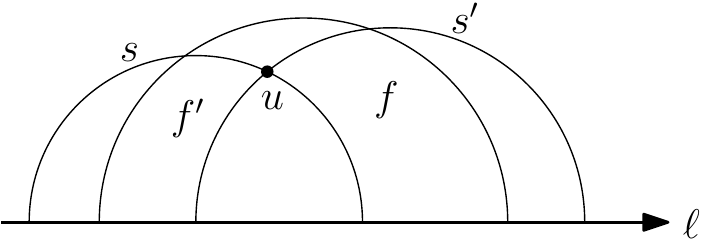}
\caption{\footnotesize Illustrating two faces $f'$ and $f$ that are adjacent in a path of $G$. The vertex $u$, which is the intersection of the boundaries of two disks $s$ and $s'$, is the rightmost vertex of $f'$ and also the leftmost vertex of $f$.
}
\label{fig:adjdisks}
\end{center}
\end{minipage}
\vspace{-0.15in}
\end{figure}

We define $S(f)$ as the subset of disks of $S$ containing $f$ and define $S(f')$ similarly. By definition, $P^*(f)$ (resp., $P^*(f')$) is the set of maximal intervals of indices of disks of $S(f)$ (resp., $S(f')$). It is easy to see that the symmetric difference of $S(f)$ and $S(f')$ comprises exactly two disks, i.e., the two disks whose boundaries intersect at $u$. As such, a straightforward analysis can prove that $|P^*(f)\setminus P^*(f')|\leq 3$ and $|P^*(f')\setminus P^*(f)|\leq 3$ as follows. Indeed, let $s$ be the disk of $S(f')\setminus S(f)$ and $s'$ be the disk of $S(f)\setminus S(f')$. Hence, $s$ contains $f'$ but not $f$ while $s'$ contains $f$ but not $f'$ (e.g., see Fig.~\ref{fig:adjdisks}). As such, comparing $P^*(f)$ to $P^*(f')$, we have the following two cases.

\begin{itemize}
  \item
  Due to that $S(f)$ ``loses'' a disk (i.e., $s$) comparing to $S(f')$, at most two new dual segments are generated in $P^*(f)$ comparing to $P^*(f')$, i.e., the interval of $S(f')$ containing the index of $s$ is divided into at most two new intervals in $P^*(f)$.

  \item
  Due to that $S(f)$ ``gains'' a disk (i.e., $s'$) comparing to $S(f')$, at most one new dual segment is generated in $P^*(f)$ in the form of one of the following three cases: (1) the index of $s'$ becomes a single interval in $P^*(f)$; (2) the index of $s'$ is merged with one interval of $P^*(f')$ to become a new interval of $P^*(f)$ with $s$ as an endpoint; (3) $s'$ is concatenated with two other intervals of $P^*(f')$ to become a new interval of $P^*(f)$ with $s'$ in the middle.
\end{itemize}

Combining the above two cases, we obtain that $|P^*(f)\setminus P^*(f')|\leq 3$. By a symmetric analysis, we can also obtain $|P^*(f')\setminus P^*(f)|\leq 3$.
\qed
\end{proof}

With the above two lemmas, we can now prove the upper bound for $P^*$.

\begin{lemma}\label{lem:boundLinf}
In the $L_2$ metric, $|P^*|=O(m+\kappa)$.
\end{lemma}
\begin{proof}
Recall that the graph $G$ consists of a set of directed paths, with the first node of each path representing an initial face. By definition, each initial face corresponds to exactly one disk of $S$. Hence, the total number of initial faces is at most $m$.
Lemmas~\ref{lem:90} and \ref{lem:100} together imply that $|P^*|\leq m+3\cdot |G|$, where $|G|$ is the number of nodes of $G$. Since the number of faces of $\calA$ is $O(m+\kappa)$, we have $|G|=O(m+\kappa)$. Therefore, we obtain $|P^*|=O(m+\kappa)$.
\qed
\end{proof}

%Combining the trivial upper bound $O(mn)$ for $|P^*|$, we have the following.
%
%\begin{corollary}
%In the $L_{\infty}$ metric, $|P^*|=O(\min\{m+\kappa,mn\})$.
%\end{corollary}

\subsection{Computing $P^*$}

We now compute $P^*$. A straightforward method is to use brute force: For each point $p_i\in P$, check the disks of $S$ one by one following their index order; in this way, $P_i^*$ can be computed in $O(m)$ time. As such, the total time for computing $P^*$ is $O(mn)$. In what follows, we present another algorithm of $O(n\log (n+m)+(m+\kappa)\log m)$ time. As discussed in Section~\ref{sec:boundL2}, it suffices to compute the dual segments generated by all faces of $\calA$ (or equivalently, generated by all nodes of the graph $G$).

%For any three consecutive nodes $f_1$, $f_2$, and $f_3$ in any path of $G$, we call $f_1$ the {\em left neighbor} of $f_2$ and $f_3$ the {\em right neighbor} of $f_2$.

The main idea of our algorithm is to directly compute for each path $\pi\in G$ the dual segments defined by the initial face of $\pi$ and then for each non-initial face $f\in \pi$, determine $P^*(f)$ indirectly based on $P^*(f')$, where $f'$ is the predecessor face of $f$ in $\pi$.
%In what follows, we first describe how to compute the segments of $P^*$ and then discuss how to determine their weights.

We begin with computing the graph $G$. To this end, we first compute the arrangement $\calA$. This can be done in $O((m+\kappa)\cdot \log m)$ time, e.g., by a line sweeping algorithm.\footnote{It might be possible to compute $\calA$ in $O(m\log m+\kappa)$ time by adapting the algorithm of~\cite{ref:ChazelleAn92segment}. However, $O((m+\kappa)\cdot \log m)$ time suffices for our purpose as other parts of the algorithm dominate the time complexity of the overall algorithm.} Then, the graph $G$ can be constructed by traversing $G$ in additional $O(m+\kappa)$ time.

Recall that we also need to determine the weight for each dual segment of $P^*$. To this end, for each face $f\in \calA$, we compute its ``weight'' that is equal to the minimum weight of all points of $P$ in $f$ (if $f$ does not contain any point of $P$, then we set its weight to $\infty$). For this, it suffices to determine the face of $\calA$ containing each point of $P$. This can be done in $O(n\log (n+m)+(m+\kappa)\log m)$ time by a line sweeping algorithm, e.g., we can incorporate this step into the above sweeping algorithm for constructing $\calA$ (alternatively one could build a point location data structure on $\calA$~\cite{ref:EdelsbrunnerOp86} and then perform point location queries for points of $P$).

We next compute $P^*(f)$ for all initial faces $f$. Consider an initial face $f$. Let $s_j$ be the disk such that $l_j$ is the leftmost vertex of $f$. According to the proof of Lemma~\ref{lem:90}, $P^*(f)$ has only one interval $[k_j,j]$ for some index $k_j\leq j$. To compute $k_j$, we can do a simple binary search on the indices in the interval $[1,j]$. Indeed, we first take $k=j/2$ and check whether $s_k$ contains $l_j$. If yes, we continue the search on $[1,k]$; otherwise we proceed on $[k,j]$. In this way, we can find $k_j$ in $O(\log m)$ time. As such, $P^*(f)$ for all initial faces $f$ can be computed in $O(m\log m)$ time.

Next, for each path $\pi$ of $G$, starting from its initial face, we compute $P^*(f)$ for all non-initial faces $f\in \pi$. Based on the analysis of Lemma~\ref{lem:100}, the following lemma shows that $P^*(f)$ can be determined in $O(\log m)$ time based on $P^*(f')$, where $f'$ is the predecessor face of $f$ in $\pi$.

\begin{lemma}\label{lem:path}
The dual segments of $\bigcup_{f\in \pi}P^*(f)$ and their weights can be computed in $O(|\pi|\cdot\log m)$ time, where $|\pi|$ is the number of nodes of $\pi$.
\end{lemma}
\begin{proof}
Let $t=|\pi|$. Let $f_1,f_2,\ldots,f_t$ be the list of nodes of $\pi$ with $f_1$ as the initial face.
Recall that $P^*(f_1)$ has exactly one interval, which has already been computed. In general, suppose $P^*(f_{i})$ has been computed. We show below that $P^*(f_{i+1})$ can be determined in $O(\log m)$ time based on the analysis of Lemma~\ref{lem:100}.

Note that intervals of $P^*(f_{i})$ are disjoint and we assume that they are stored in a balanced binary search tree $T(f_i)$ sorted by the left endpoints of the intervals. Since $|P^*(f_i)|\leq \lceil m/2\rceil$, the size of $T(f_i)$ is $O(m)$.
Let $u$ be the rightmost vertex of $f_i$, which is also the leftmost vertex of $f_{i+1}$. Let $s_j$ and $s_k$ be the two disks that intersect at $u$ such that $s_j$ contains $f_{i+1}$ but not $f_i$ (e.g., see Fig.~\ref{fig:adjdisksalg}). Hence, $s_k$ contains $f_i$ but not $f_{i+1}$. As discussed in the proof Lemma~\ref{lem:100}, there are two cases that lead to changes from $P^*(f_i)$ to $P^*(f_{i+1})$.

\begin{figure}[t]
\begin{minipage}[t]{\textwidth}
\begin{center}
\includegraphics[height=1.0in]{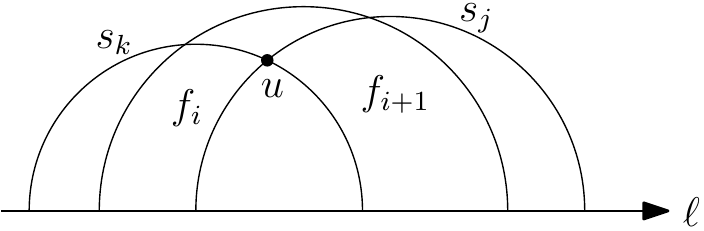}
\caption{\footnotesize Illustrating the two faces $f_i$ and $f_{i+1}$ as well as the two disks $s_k$ and $s_j$.
}
\label{fig:adjdisksalg}
\end{center}
\end{minipage}
\vspace{-0.15in}
\end{figure}

\begin{itemize}
  \item Due to that $s_k$ contains $f_i$ but not $f_{i+1}$, we first find the interval of $P^*(f_i)$ containing the index $k$, which can be done in $O(\log m)$ time using the tree $T(f_i)$. Then, we remove $k$ from the interval, which splits the interval into two new intervals (degenerate case happens if $k$ is an endpoint of the interval, in which case only one new interval is produced and that interval could be empty as well; we only discuss the non-degenerate case below as the degenerate case can be handled similarly). We remove the original interval from $T(f_i)$ and then insert the two new intervals into $T(f_i)$.
  \item Due to that $s_j$ contains $f_{i+1}$ but not $f_i$, we need to add the index $j$ to the intervals of $P^*(f_i)$ in order to obtain $P^*(f_{i+1})$. To this end, we find the two intervals of $P^*(f_i)$ closest to $j$, one on the left side of $j$ and the other on the right side of $j$; this can be done in $O(\log m)$ time using the tree $T(f_i)$. As discussed in the proof Lemma~\ref{lem:100}, depending on whether $j$ is adjacent to one, both, or neither of the two intervals, we will update $T(f_i)$ accordingly (more specifically, at most two intervals are removed from $T(f_i)$ and exactly one interval is inserted into $T(f_i)$).
\end{itemize}

The above performs $O(1)$ insertion/deletion operations on $T(f_i)$, which together take $O(\log m)$ time. The resulting tree is $T(f_{i+1})$, representing all intervals of $P^*(f_{i+1})$. In addition, once an interval is removed from the tree, we add the interval to $\calI$ (which is $\emptyset$ initially). After the last face $f_t$ is processed, we obtain $T(f_t)$, representing $P^*(f_t)$. We then add all intervals of $T(f_t)$ to $\calI$, after which $\calI$ is $\bigcup_{f\in \pi}P^*(f)$.

The above computes all dual segments of $\calI=\bigcup_{f\in \pi}P^*(f)$ in $O(|\pi|\cdot \log m)$ time. However, we also need to determine the weights of these segments. To this end, we modify the above algorithm as follows.

We build a data structure on the weights of the nodes of $\pi$ to support the following {\em range-minima query}: Given a range $[i,j]$ with two indices $1\leq i\leq j\leq t$, the query asks for the minimum weight of all faces $f_k$ with $k\in [i,j]$. We can easily achieve $O(\log t)$ query time by constructing in $O(t)$ time an augmenting binary search tree on the weights of $f_1,f_2,\ldots,f_t$.\footnote{It is possible to achieve $O(1)$ time query with $O(t)$ preprocessing time using the data structures of~\cite{ref:BenderTh00,ref:HarelFa84}; however the simple binary search tree solution with $O(\log m)$ query time suffices for our purpose.} Note that $\log t=O(\log m)$ since $t=O(m^2)$.

For each interval $I\in \calI$, when it is first time inserted into $T(f_i)$ for some face $f_i$, we set $a(I)=i$. When $I$ is deleted from $T(f_k)$ for some face $f_k$, we know that all faces $f_i,f_{i+1},\ldots,f_k$ define $I$ (i.e., $I$ is in $P^*(f_h)$ for all $i\leq h\leq k$) and thus the weight of $I$ is equal to the minimum weight of these faces; to find the minimum weight, we perform a range-minimum query with $[a(I),k]$ in $O(\log m)$ time. As such, this change introduces a total of $O(|\calI|\cdot \log m)$ additional time to the overall algorithm.
Therefore, the overall time of the entire algorithm is still bounded by $O(|\calI|\cdot \log m)$ time, which is $O(|\pi|\cdot \log m)$ as $|\bigcup_{f\in \pi}P^*(f)|=O(|\pi|)$ by Lemmas~\ref{lem:90} and \ref{lem:100}.

The lemma thus follows
\qed
\end{proof}

We apply the algorithm of Lemma~\ref{lem:path} to all paths of $G$, which takes $O(|G|\cdot \log m)$ time in total. After that, all dual segments of $P^*$ with their weights are computed. Recall that $|G|=O(m+\kappa)$. Hence, the time of the overall algorithm for computing $P^*$ is bounded by $O(n\log(n+m)+(m+\kappa)\log m)$.
Consequently, using the dual transformation, we can solve the $L_2$ hitting set problem on $P$ and $S$. The following theorem analyzes the time complexity of the overall algorithm.

\begin{theorem}\label{theo:l2}
The line-constrained $L_2$ hitting set problem can be solved in $O((n+m)\log (n+m)+\kappa\log m)$ time, where $\kappa$ is the number of pairs of disks that intersect.
\end{theorem}
\begin{proof}
As discussed above, computing $P^*$ takes $O(n\log(n+m)+(m+\kappa)\log m)$ time. As $|S^*|=n$ and $|P^*|=O(m+\kappa)$ by Lemma~\ref{lem:boundLinf}, applying the 1D dual coverage algorithm in~\cite{ref:PedersenAl22} takes $O((|S^*|+|P^*|)\log (|S^*|+|P^*|))$ time, which is $O((n+m+\kappa)\log (n+m+\kappa))$.

We claim that $(n+m+\kappa)\log (n+m+\kappa)=O((n+m)\log (n+m)+\kappa\log m)$. Indeed, since $\kappa=O(m^2)$, it suffices to show that $\kappa\log (n+m)=O((n+m)\log (n+m)+\kappa\log m)$. If $n<m^2$, then $\log (n+m)=O(\log m)$ and thus $\kappa\log (n+m)=O((n+m)\log (n+m)+\kappa\log m)$ holds; otherwise, we have $\kappa \leq m^2\leq n$ and thus $\kappa\log (n+m)=O((n+m)\log (n+m)+\kappa\log m)$ also holds.

As such, the total time of the overall algorithm for solving the $L_2$ hitting set problem is bounded by $O((n+m)\log (n+m)+\kappa\log m)$.\qed
\end{proof}

Recall that $P^*$ can also be computed in $O(mn)$ time by a straightforward brute force method. Using the dual transformation, the $L_2$ problem can also be solved in $O(nm\log (n+m))$ time. This algorithm may be interesting when $n$ is much smaller than $m$.

\section{The line-separable unit-disk hitting set and the half-plane hitting set}
\label{sec:line-separable}

In this section, we demonstrate that our techniques for the line-constrained disk hitting set problem can be utilized to solve other geometric hitting set problems.

\paragraph{\bf Line-separable unit-disk hitting set.}
We first consider the line-separable unit-disk hitting set problem, in which $P$ and centers of $S$ are separated by a line $\ell$ and all disks of $S$ have the same radius. Without loss of generality, we assume that $\ell$ is the $x$-axis and all points of $P$ are above (or on) $\ell$. 
%As such, the portion of $s_i$ above $\ell$ is a subset of its upper half disk. 
Since disks of $S$ have the same radius and their centers are below (or on) $\ell$, the boundaries of every two disks intersect at most once above $\ell$ (referred to as the {\em single-intersection} property). 
%We define $\kappa$ as the number of pairs of disks whose boundaries intersect above $\ell$. 
Due to the single-intersection property, to solve the problem, we can simply use the same algorithm as in Section~\ref{sec:l2} for the line-constrained $L_2$ case. Indeed, one can verify that the following lemmas that the algorithm relies on still hold: Lemmas~\ref{lemma:80}, \ref{lem:90}, \ref{lem:100}, \ref{lem:boundLinf}, and \ref{lem:path}. By Theorem~\ref{theo:l2} (and the discussion after it), we obtain the following result.

\begin{theorem}\label{theo:line-sep}
Given in the plane a set $P$ of $n$ weighted points and a set $S$ of $m$ unit disks such that $P$ and centers of disks $S$ are separated by a line $\ell$, one can compute a minimum weight hitting set of $P$ for $S$ in $O(nm\log(m+n))$ time or in $O((n+m)\log(n+m)+\kappa\log m)$ time,
where $\kappa$ is the number of pairs of disks of $S$ whose boundaries intersect in the side of $\ell$ containing $P$.
\end{theorem}

\paragraph{\bf Remark.}
Although disks of $S$ have the same radius, since their centers may not be on the same line, one can verify that Lemma~\ref{lem:20} does not necessarily hold any more. Consequently, the algorithm in Section~\ref{sec:unitDist} for the line-constrained unit-disk case cannot be applied in this scenario. 
%However, if the centers of disks of $S$ all lie on the same line parallel to $\ell$ (and below $\ell$), then Lemma~\ref{lem:20} still holds true. As such, we can still employ  the same algorithm as described in Section~\ref{sec:unitDist} to solve the problem in $O((n+m)\log (n+m))$ time.
%\bigskip

\paragraph{\bf Half-plane hitting set.}
In the half-plane hitting set problem, we are given in the plane a set $P$ of $n$ weighted points and a set $S$ of $m$ half-planes. The goal is to compute a subset of $P$ of minimum weight so that every half-plane of $S$ contains at least one point in the subset. 
In the {\em lower-only case}, all half-planes of $S$ are lower half-planes. 

The lower-only case problem can be reduced to the line-separable unit-disk hitting set problem, as follows. We first find a horizontal line $\ell$ below all points of $P$. Then, since each half-plane $h$ of $S$ is a lower one, $h$ can be considered as a disk of infinite radius with center below $\ell$. As such, $S$ becomes a set of unit disks with centers below $\ell$. By Theorem~\ref{theo:line-sep}, we have the following result.\footnote{Another way to see this is the following. The main property our algorithm for Theorem~\ref{theo:line-sep} relies on is the single-intersection property, that is, the boundaries of any two disks intersect at most once above $\ell$. This property certainly holds for the half-planes of $S$ and thus the algorithm is applicable.}

\begin{theorem}\label{theo:loweronly}
Given in the plane a set $P$ of $n$ weighted points and a set $S$ of $m$ lower half-planes, one can compute a minimum weight hitting set of $P$ for $S$ in $O(nm\log(m+n))$ time or in $O(n\log n+ m^2\log m)$ time.
\end{theorem}

As discussed in Section~\ref{sec:intro}, using duality to reduce the problem to the lower-only case half-plane coverage problem and applying the coverage algorithm in~\cite{ref:PedersenAl22}, one can solve the lower-only case half-plane hitting set problem in $O(m\log m+n^2\log n)$ time. Combining this result with Theorem~\ref{theo:loweronly} leads to the following. 

\begin{corollary}\label{coro:halfplane}
Given in the plane a set $P$ of $n$ weighted points and a set $S$ of $m$ lower half-planes, one can compute a minimum weight hitting set of $P$ for $S$ in $O((n+m)\log (n+m)+ k^2\log k)$ time, where $k=\min\{m,n\}$. 
\end{corollary}

For the general case where $S$ contains both lower and upper half-planes, we show that the problem can be reduced to $O(n^2)$ instances of the lower-only case problem, as follows.

We first discuss some observations on which our algorithm relies. 
Consider an optimal solution $P_{opt}$, i.e., a minimum weight hitting set of $P$ for $S$. Let $\calH$ denote the convex hull of $P_{opt}$. 
Let $p$ and $q$ be the leftmost and rightmost vertices of $\calH$, respectively. 
Let $\calH_1$ (resp., $\calH_2$) denote the set of vertices of the lower (resp., upper) hull of $\calH$ excluding $p$ and $q$. As such, $\calH_1$, $\calH_2$, and $\{p,q\}$ form a partition of the vertex set of $\calH$.
%Let $\ell_{pq}$ denote the line through $p$ and $q$. 
Define $P^1_{pq}$ (resp., $P^2_{pq}$) to be the subset of points of $P$ below (resp., above) the line through $p$ and $q$.
Denote by $S^0_{pq}$ the subset of half-planes of $S$ each of which is hit by either $p$ or $q$. Let $S^1_{pq}$ (resp., $S^2_{pq}$) be the subset of lower (resp., upper) half-planes of $S\setminus S^0_{pq}$. As such, $S^0_{pq}$, $S^1_{pq}$, and $S^2_{pq}$ form a partition of $S$. Observe that each (lower) half-plane of $S^1_{pq}$ is hit by a point of $\calH_1$ but not hit by any point of $\calH_2$, and each (upper) half-plane of $S^2_{pq}$ is hit by a point of $\calH_2$ but not hit by any point of $\calH_1$~\cite{ref:Har-PeledWe12}. As such, we further have the following observation. 

\begin{observation}\label{obser:halfplane}
For each $i=1,2$, $\calH_i$ is an optimal solution to the half-plane hitting set problem for $P^i_{pq}$ and $S^i_{pq}$. 
\end{observation}

Note that the half-plane hitting set problem for each $i=1,2$ in Observation~\ref{obser:halfplane} is an instance of the lower-only case problem. 

In light of Observation~\ref{obser:halfplane}, our algorithm for the hitting set problem for $P$ and $S$ works as follows. For any two points $p$ and $q$ of $P$, we do the following. %Define $\ell_{pq}$, $P^1_{pq}$, $P^2_{pq}$, $S^0_{pq}$, $S^1_{pq}$, and $S^2_{pq}$, as above. 
Following the definitions as above, we first compute $P^1_{pq}$, $P^2_{pq}$, $S^0_{pq}$, $S^1_{pq}$, and $S^2_{pq}$, which takes $O(n+m)$ time. Then, for each $i=1,2$, we solve the lower-only case half-plane hitting set problem for $P^i_{pq}$ and $S^i_{pq}$, and let $P^i_{\text{opt}}$ denotes the optimal solution. We keep $P^1_{\text{opt}}\cup P^2_{\text{opt}}\cup \{p,q\}$ as a candidate solution for our original hitting set problem for $P$ and $S$. In this way, we reduce our hitting set problem for $P$ and $S$ to $O(n^2)$ instances of the lower-only case half-plane hitting set problem (each instance involves at most $n$ points and at most $m$ half-planes). Among all $O(n^2)$ candidate solutions, we finally return the one of minimum weight as an optimal solution. The total time of the algorithm is bounded by $O(n^2\cdot (n+m+T))$, where $T$ is the time for solving the lower-only case hitting set problem for at most $n$ points and at most $m$ half-planes. Using Corollary~\ref{coro:halfplane}, we obtain the following result. 

\begin{theorem}\label{theo:halfplane}
Given in the plane a set $P$ of $n$ weighted points and a set $S$ of $m$ half-planes, one can compute a minimum weight hitting set of $P$ for $S$ in $O(n^2(n+m)\log (n+m)+ n^2k^2\log k)$ time, where $k=\min\{m,n\}$. 
\end{theorem}

When $m=n$, the runtime of our algorithm is $O(n^4\log n)$, which improves the previous best result of $O(n^6)$ time~\cite{ref:Har-PeledWe12} by nearly a quadratic factor. 

\section{Concluding remarks}
\label{sec:conclusion}

In this paper, we solve the line-constrained disk hitting set problem in $O((m+n)\log(m+n)+\kappa\log m)$ time in the $L_2$ metric, where $\kappa$ is the number of pairs of disks that intersect. The factor $\kappa\log m$ can be removed for the 1D, $L_1$, $L_{\infty}$, and unit-disk cases. An alternative (and relatively straightforward) algorithm also solves the $L_2$ case in $O(nm\log (n+m))$ time. Our techniques can also be used to solve other geometric hitting set problems. 

We can prove an $\Omega((n+m)\log (n+m))$ time lower bound for the problem even for the 1D unit-disk case (i.e., all segments have the same length), by a simple reduction from the element uniqueness problem (Pedersen and Wang~\cite{ref:PedersenAl22} used a similar approach to prove the same lower bound for the 1D coverage problem). Indeed, the element uniqueness problem is to decide whether a set $X=\{x_1,x_2,\ldots,x_N\}$ of $N$ numbers are distinct. We construct an instance of the 1D unit-disk hitting set problem with a point set $P$ and a segment set $S$ on the $x$-axis $\ell$ as follows. For each $x_i\in X$, we create a point $p_i$ on $\ell$ with $x$-coordinate equal to $x_i$ and create a segment on $\ell$ that is the point $p_i$ itself. Let $P=\{p_i \ |\ 1\leq i\leq N\}$ and $S$ the set of segments defined above (and thus all segments have the same length); then $|P|=|S|=N$. We set the weights of all points of $P$ to $1$.  Observe that the elements of $X$ are distinct if and only if the total weight of points in an optimal solution to the 1D unit disk hitting set problem on $P$ and $S$ is $n$. As the element uniqueness problem has an $\Omega(N\log N)$ time lower bound under the algebraic decision tree model, $\Omega((n+m)\log (n+m))$ is a lower bound for our 1D unit disk hitting set problem.

The lower bound implies that our algorithms for the 1D, $L_1$, $L_{\infty}$, and unit-disk cases are all
optimal. It would be interesting to see whether faster algorithms exist for the $L_2$ case or some non-trivial lower bounds can be proved (e.g., 3SUM-hard~\cite{ref:GajentaanOn95}).

%\section*{Acknowledgment}


%%%%%%%%%%%%
\bibliographystyle{plain}
%\bibliographystyle{splncs}

%%%%%%%%%%%%%%%%%%%%%%%%%%%%%%%%%%%%%%%%%%%%%%%%%%%%%%%%%%%%%
%% APPENDICES
%%%%%%%%%%%%%%%%%%%%%%%%%%%%%%%%%%%%%%%%%%%%%%%%%%%%%%%%%%%%%
%\appendix

%\newpage
%\appendix
%\section*{APPENDIX}

%\vspace{0.1in}
%\noindent
%{\bf Lemma \ref{lem:10}.}
%{\em
%For any $i\in [1,n-1]$, $I_i(\alpha)\geq I_{i+1}(\alpha)$,
%$I_i(\beta)\geq I_{i+1}(\beta)$,
%$I_i(\gamma)\leq I_{i+1}(\gamma)$, and $I_i(\delta)\leq I_{i+1}(\delta)$  (e.g., see Fig.~\ref{fig:shift}).
%}
%\vspace{0.06in}

%\begin{figure}[t]
%\begin{minipage}[t]{\textwidth}
%\begin{center}
%\includegraphics[height=1.2in]{computeIbeta.eps}
%\caption{\footnotesize Illustrating the path (the dotted curve) from $v_n$ to $v_{k-1}$
%using the edge $e(i,j)$.}
%\label{fig:computeIbeta}
%\end{center}
%\end{minipage}
%\vspace{-0.15in}
%\end{figure}

\end{document}